\documentclass[preprint,12pt]{elsarticle}
\usepackage{hyperref}
\usepackage{times}
\usepackage{amscd,amsmath,amssymb,amsthm}
\usepackage{color} 

\usepackage{mathrsfs}
\usepackage{graphicx}

\usepackage{helvet}
\usepackage{courier}

\usepackage{algorithm}
\usepackage[noend]{algorithmic}
\usepackage{verbatim}
\usepackage{flafter}

\newcommand{\rw}{\rightarrow}

\newcommand{\ds}{\delta^{\ast}}
\newcommand{\sa}{\Sigma^{\ast}}

\newcommand{\allseries}{S\langle\!\langle\Sigma^{\ast}\rangle\!\rangle}

\theoremstyle{definition}
\newtheorem{definition}{Definition}
\theoremstyle{plain}
\newtheorem{proposition}{Proposition}
\newtheorem{theorem}{Theorem}
\newtheorem*{thm-other}{Theorem}
\newtheorem{example}{Example}
\newtheorem{corollary}{Corollary}
\newtheorem{lemma}{Lemma}
\theoremstyle{remark}
\newtheorem{remark}{Remark}

\begin{document}

\journal{arXiv}

\begin{frontmatter}


\title{On Quotients of Formal Power Series
\thanks{This work is supported by National Science Foundation of China (Grant No: 60873119) and the Higher School Doctoral Subject Foundation of Ministry of Education of China (Grant No:200807180005).}}

\author[1]{Yongming Li\corref{cor1}}
\ead{liyongm@snnu.edu.cn}

\author [1]{Qian Wang}

\address[1]{College of Computer Science, Shaanxi Normal University, Xi'an, 710062, China}

\author[2]{Sanjiang Li}
\ead{sanjiang.li@uts.edu.au}

\address[2]{Centre for Quantum Computation and Intelligent Systems,
       Faculty of Engineering and Information Technology, University of Technology
       Sydney, Australia}

\cortext[cor1]{Corresponding Author}
\begin{abstract}
Quotient is a basic operation of formal languages, which plays a key role in the construction of minimal deterministic finite automata (DFA) and the universal automata. In this paper, we extend this operation to formal power series and systemically investigate its implications in the study of weighted automata. In particular, we define two quotient operations for formal power series that coincide when calculated by a word. We term the first operation as (left or right) \emph{quotient}, and the second as (left or right) \emph{residual}. To support the definitions of quotients and residuals, the underlying semiring is restricted to complete semirings or complete c-semirings. Algebraical properties that are similar to the classical case are obtained in the formal power series case. Moreover, we show closure properties, under quotients and residuals, of regular series and weighted context-free series are similar as in formal languages. Using these operations, we define for each formal power series $A$ two weighted automata ${\cal M}_A$ and ${\cal U}_A$. Both weighted automata accepts $A$, and ${\cal M}_A$ is the minimal deterministic weighted automaton of $A$. The universality of  ${\cal U}_A$ is justified and, in particular, we show that  ${\cal M}_A$  is a sub-automaton of ${\cal U}_A$.  Last but not least, an effective method to construct the universal automaton is also presented in this paper.
\end{abstract}


\end{frontmatter}

\section{Introduction}

In formal language theory, quotient is a basic and very important operation and plays a fundamental role in the construction of minimal deterministic finite automata (DFA). Given a formal language $L$ over an alphabet $\Sigma$, the left quotient $u^{-1}L$ of $L$ by a word $u$ is defined as the language $\{v\in \Sigma^\ast | uv\in L\}$, where $\Sigma^\ast$ is the free monoid of words over $\Sigma$. The famous Myhill-Nerode Theorem then states that $L$ is a  regular language if and only if the number of different left quotients of $L$ (also called the quotient complexity \cite{Brz09} of $L$) is finite. Moreover, a minimal DFA which recognizes $L$ can be constructed in a natural way by using left quotients as states. In particular, this means that the quotient complexity of $L$ is equal to the size of the minimal DFA which recognizes $L$.

The notion of left quotient of a formal language by a word  can be extended to quotients by a formal language in two ways. Given two formal languages $L,X$, the left quotient of $L$ by $X$, denoted by $X^{-1}L$, is defined as the union of $u^{-1}L$ for all words $u$ in $X$.  Another extension is less well-known, if not undefined at all. We define the \emph{left residual} of $L$ by $X$, denoted by $X\backslash L$, as the intersection of $u^{-1}L$ of all words in $X$.  Similarly we have $LX^{-1}$, the right quotient of $L$ by $X$, and $L/X$, the right residual of $L$ by $X$. Regarding each left residual of $L$ as a state, there is a natural way to define an automaton, which is called the \emph{universal automaton} \cite{conway71,sakarovitch09} of $L$. The universal automaton of a formal language $L$ contains many interesting information (e.g. factoraization) of $L$ \cite{lombardy08} and plays a very important role in constructing the minimal nondeterministic finite automaton (NFA) of $L$ \cite{ADN92,polak05}.


Former power series are extensions of formal languages, which are used to describe the behaviour of weighted automata (i.e. finite automata with weights). Weighted automata were introduced in 1961 by Sch\"{u}tzenberger in his seminal paper \cite{schutzenberger61}.  A formal power series is a mapping from $\Sigma^\ast$, the free monoid of words over $\Sigma$, into a semiring $S$. Depending on the choice of the semiring $S$, formal power series can be viewed as weighted, multivalued or quantified languages where each word is assigned a weight, a number, or some quantity. Weighted automata have been used to describe quantitative properties in areas such as probabilistic systems, digital image compression, natural language processing.  We refer to \cite{droste09} for an detailed introduction of weighted automata and their applications.

Despite that a very large amount of work has been devoted to the study of formal power series and weighted automata (see e.g.  \cite{kuich86, salomaa78, berstel11, droste09,esik10} for surveys), the important concept of quotient as well as universal automata has not been systematically investigated in this weighted context. The only exception seems to be \cite{berstel11}, where the quotient of formal power series (by word) was discussed in  pages 10-11. When the semiring is complete, it is straightforward to extend the definition of the quotient of a formal power series $A$ from words to series: we only need to take the weighted sum of all left quotients of $A$ by word in $\Sigma$. Our attempt to characterize the residual of a formal power series $A$ by a formal power series as the weighted intersection of all left quotients of $A$ by word in $\Sigma$ is, however, unsuccessful. Several important and nice properties fail to hold anymore.

The aim of this paper is to introduce the quotient and residual operations in formal power series and study their application in the minimization of weighted automata. To overcome the above obstacle with residuals, we require the semiring to be a complete c-semiring (to be defined in Section~2), and then give a characterization of residuals in terms of quotients by word. Many nice properties and useful notions then follow in a natural way.


The remainder of this paper is organized as follows. Section 2 introduces basic notions and properties of semirings, formal power series, and weighted automata. Quotients of formal power series are introduced in Section 3, where we also show how to construct the minimal deterministic weighted automata effectively. In Section 4, we introduce the residuals and factorizations of formal power series. Using the left residuals, we define the universal weighted automaton $\mathcal{U}_A$ for arbitrary formal power series $A$ in Section~5, and justify its universality in Section~6.
An effective method for constructing the universal automaton is described in Section~7, which is followed by a comparison of the quotient and the residual operations. The last section concludes the paper.

\section {Preliminaries}

We recall in this section the notions of semirings, formal power series, weighted automata, and weighted contex-free grammar. Interested readers are referred to  \cite{droste09,kuich86,salomaa78,esik10}  for more information.

\subsection{Semirings}

 A 5-tuple ${\cal S}=(S, \oplus, \otimes, 0,1)$ is called a \emph{ semiring} if $S$ is a set containing at least two different elements $0$ and $1$, and $\oplus $ and $\otimes$ are two binary operations on $S$ such that
\begin{itemize}
\item [(i)] $\oplus $ is associative and is commutative and has identity $0$;
\item [(ii)] $\otimes$ is associative and has identity $1$ and null element $0$ (i.e., $a\otimes 0=0\otimes a=0$ for all $a\in S$); and
\item [(iii)] $\otimes$ distributes over $\oplus $, i.e., for all $a,b,c\in S$, $a\otimes (b\oplus c)=(a\otimes b)\oplus (a \otimes c)$ and $(b\oplus c)\otimes a=(b\otimes a)\oplus (c\otimes a)$.
\end{itemize}

Intuitively, a semiring is a ring (with unity) without subtraction. All rings (with unity), as well as all fields, are
semirings, e.g., the integers $\mathbb{Z}$, rationals $\mathbb{Q}$, reals $\mathbb{R}$, complex numbers $\mathbb{C}$. Lattices provide another important type of semirings. Recall that a partially ordered set $(L, \leq)$ is a \emph{lattice} if for any two elements
$a, b\in L$, the least upper bound $a\vee b = \sup\{a, b\}$ and the greatest lower
bound $a\wedge b = \inf\{a, b\}$ exist in $(L,\leq)$. A lattice $(L, \leq)$ is \emph{distributive}, if
$a\wedge(b\vee c) = (a\wedge b)\vee (a\wedge c)$ for all $a, b, c\in L$; and \emph{bounded}, if $L$ contains a
smallest element, denoted 0, and a greatest element, denoted 1. Let $(L, \leq)$ be any bounded distributive lattice. Then $(L,\vee,\wedge, 0, 1)$ is a semiring. Because
a distributive lattice $L$ also satisfies the dual distributive law $a\vee(b\wedge c) = (a\vee b)\wedge (a\vee c)$ for
all $a, b, c\in L$, the structure $(L,\wedge,\vee, 1, 0)$ is also a semiring.

Other important examples of semirings include:
\begin{itemize}
\item [-] The Boolean semiring $\mathbb{B} =( \{0, 1\}, \vee, \wedge, 0, 1)$;
\item [-] The semiring of the natural numbers $(\mathbb{N}, +, \cdot, 0, 1)$ with the usual addition and multiplication;
\item [-] The tropic semiring $(\mathbb{N}\cup\{\infty\}, \min, +,\infty, 0)$ with min and $+$ extended to $\mathbb{N}\cup\{\infty\}$ in a natural way;
\item [-] The min-sum semiring of nonnegative reals $(\mathbb{R}^+\cup\{0, \infty\}, \min, +, \infty, 0)$;
\item [-] The semiring of (completely positive) super-operators on a Hilbert space $\cal H$ $({\cal SO}({\cal H}), +, \circ, 0_{\cal H}, {\cal I}_{\cal H})$.
\end{itemize}
We note in the last semiring the addition is not idempotent and the product is not commutative. This semiring is recently used in model checking quantum Markov chains \cite{Feng+12} and the study of finite automata with weights taken from this semiring just initiated.

\subsubsection{Complete Semiring}

Let $I$ be an index set and let $S$ be a semiring. An infinitary sum operation $\sum_I : S^I \rw S$ is an operation that associates with every family $\{a_i | i\in I\}$ of elements of $S$ an element $\sum_{i\in I} a_i$ of $S$. A semiring $S$ is called \emph{complete} if it has an infinitary sum operation $\sum_I$
for each index set $I$ and the following conditions are satisfied:
\begin{itemize}
\item [(i)] $\sum_{i\in\emptyset}a_i= 0$, $\sum_{i\in\{j\}}a_i=a_j$, and $\sum_{i\in\{j,k\}}a_i=a_j\oplus a_k$ for $j\not=k$.

\item [(ii)] $\sum_{j\in J}\sum_{i\in I_j}a_i=\sum_{i\in I}a_i$ if $\bigcup_{j\in J}I_j=I$ and $I_j\cap I_k=\emptyset$ for $j\not=k$.

\item [(iii)] $\sum_{i\in I}(a\otimes a_i) = a \otimes\sum_{i\in I} a_i$ and $\sum_{i\in I}(a_i\otimes a) = \sum_{i\in I} a_i \otimes a$.
\end{itemize}

This means that a semiring $S$ is complete if it is possible to define infinite sums (i) that are extensions of the finite sums, (ii) that are associative and
commutative, and (iii) that satisfy the distributive laws.

\subsubsection{Complete c-Semiring}

A semiring $S$ is a \emph{c-semiring} if $\oplus $ is idempotent (i.e., $a\oplus a=a$ for all $a\in S$), $\otimes $ is commutative, and $1$ is the absorbing element of $\oplus $ (i.e., $a\oplus 1=1$ for any $a\in S$). In general, for a semiring $S$, we define a preorder $\leq_S$ over the set $S$ by
 \begin{equation*}
\mbox{$a\leq _S b$ iff $a\oplus c=b$ for some $c\in S$.}
\end{equation*}
If $\oplus $ is idempotent, then $\leq_S$ is also a partial order.  Suppose $S$ is a c-semiring. For any $a,b\in S$, we have $0\leq_S a\leq_S 1$ and $a\oplus b = a\vee b$ (the least upper bound of $a$ and $b$) in the poset $(S,\leq_S)$. If $S$ is clear from the context, then $S$ is omitted and we simply write $\leq$ for $\leq_S$ in the following.

A semiring $S$ is called a \emph{complete c-semiring} if $S$ is a complete semiring and a c-semring. In a complete c-semiring, the infinitary sum $\sum_{i\in I} a_i$ is exactly the least upper bound of $a_i$ ($i\in I$) in $S$ under the induced partial order $\leq_S$. In this case, $\sum_{i\in I} a_i$ is also written as $\bigvee_{i\in I}a_i$.

Complete c-semiring is a special kind of the notion of \emph{quantale} \cite{rothal}, which is a complete lattice $L$ equipped with a multiplication operator $\otimes$ such that $(L,\otimes)$ is a semigroup satisfying the following distributive laws:

\begin{equation}
\label{eq: quantale-dist}
\bigvee_{i\in I}(a\otimes a_i) = a \otimes\bigvee_{i\in I} a_i,\  \ \bigvee_{i\in I}(a_i\otimes a) = \bigvee_{i\in I} a_i \otimes a.
\end{equation}

Since the infinite distributive laws (Eq.~\ref{eq: quantale-dist}) holds, there are two \emph{adjunctions} or \emph{residuals}, denoted $a\backslash b$ (left residual) and $b/a$ (right residual), respectively, satisfying the following adjunction (residual) conditions,

\begin{equation}
\label{eq:quantale-adj}
x\leq a\backslash b\  \mbox{iff} \  a\otimes x\leq b,\ \mbox{and}\  x\leq b/a\ \mbox{iff}\  x\otimes a\leq b
\end{equation}

When the given quantale is commutative, i.e., the operation $\otimes$ is commutative, then the left residual is the same as the right residual. In this case, we call the left residual (and the right residual) the residual, denoted by $a\rw b$. Then we have
 \begin{equation}
 \label{a-rw-b}
 a\rw b = \bigvee \{x | a\otimes x \leq b\}.
 \end{equation}
 We have the following proposition.
\begin{proposition}
Let $S$ be a complete c-semiring. Then $S$  is a commutative quantale with unit $1$ as the largest element of $S$.
\end{proposition}

\subsection{Formal Power Series}

Let $\Sigma$ be an alphabet and $S$ a semiring. Write $\sa$ for the set of all finite strings (or words) over $\Sigma$, and write $\varepsilon$ for the empty string. Then $\sa$ is the free monoid generated by $\Sigma$ under the operation of concatenation. We write $\Sigma^+$ for all finite non-empty strings over $\Sigma$.
A \emph{formal power series} $A$ is a mapping from $\sa$ into $S$. For simplicity, we also call a formal power series as a \emph{series}, or an \emph{$S$-subset} of $\sa$. The value of $A$ at a word $w\in\sa$ is denoted $(A,w)$ or $A(w)$ in this paper. We write $A$ as a formal sum
\begin{equation}\label{eq:formal-sum}
A=\sum_{w\in\sa}(A,w)w,
\end{equation}
where the values $(A,w)$ are referred as the \emph{coefficients} of $A$. The collection of all power series $A$ as defined above is denoted by $\allseries$. For a series $A$, if $A(\varepsilon)=0$, then $A$ is called \emph{proper}. For any series $A$, $A$ can be written as the sum of a proper series and and non-proper series, i.e.,
\begin{equation}
\label{eq:Aepsilon}
A=(A,\varepsilon)\varepsilon+\sum_{w\in \Sigma^+}(A,w)w.
\end{equation}

For s series $A$ on $\Sigma$, the \emph{support} of $A$ is defined as
\begin{equation}
\label{eq:supp}
supp(A)=\{w\in\sa | (A,w)\not=0\}.
\end{equation}

For two series $A$ and $B$ in $\allseries$, we define $A\leq B$ whenever $A(w)\leq_S B(w)$ for any $w\in\sa$.

\subsection{Weighted Automata}
Weighted automata are an extension of the classical finite automata.
\begin{definition}\label{def:l-vfa}
Let $S$ be a semiring. A \emph{weighted automaton} with weights in $S$ is a
5-tuple ${\cal A}=(Q,\Sigma,\delta,I,F)$, where $Q$ denotes a
set of states, $\Sigma$ is an input alphabet, $\delta$ is a mapping
from $Q\times \Sigma\times Q$ to $S$, and $I$ and $F$ are two mappings from $Q$ to $S$. We call ${\cal A}$  a \emph{finite} weighted automaton if both $Q$ and $\Sigma$ are finite sets; We call ${\cal A}$  a \emph{deterministic weighted automaton} (DWA for short) if $\delta$ is crisp and deterministic, i.e., $\delta$ is a mapping from $Q\times \Sigma$ into $Q$, and $I=q_0\in Q$.

The mapping $\delta$ is called the \emph{weighted (state) transition relation}. Intuitively, for
any $p,q\in Q$ and $\sigma\in \Sigma$, $\delta(p,\sigma,q)$ stands
for the weight that input $\sigma$ causes state $p$ to become $q$. $I$ and $F$ represent the (weighted) initial and, respectively, final state. For each $q\in Q$, $I(q)$ indicates the weight that $q$ is an initial state, $F(q)$ expresses the weight that $q$ is a final state.
 \end{definition}

 \begin{remark}
 Our definition of deterministic weighted automaton is different from the one used in e.g. \cite{buchsbaum00}, where a weighted automaton ${\cal A}=(Q,\Sigma,\delta,I,F)$ is \emph{deterministic} or \emph{sequential} if there exists a unique state $q_0$ in $Q$ such that $I(q_0)\not=0$ and for all $q\in Q$ and all $\sigma\in\Sigma$, there is at most one $p\in Q$ such that $\delta(q,\sigma,p)\not=0$.

These two definitions are not identical in general.
In fact, deterministic weighted automaton called in this paper is just the \emph{simple} deterministic weighted automaton defined  in \cite{buchsbaum00}, for the detail comparison, we refer to \cite{buchsbaum00,droste09}. A simple deterministic weighted automaton is obviously a sequential weighted automaton defined in \cite{buchsbaum00}. The following example shows that the converse does not hold in general.  Let $S=(\mathbb{N}\cup\{\infty\}, \min, +,\infty, 0)$ be the tropical semiring. Suppose $\Sigma=\{a\}$ and $A$ is the formal power series over $\Sigma$ defined by $A(w)=|w|$, where $|w|$ denotes the length of the string $w$. Then $A$ can not be accepted by any simple deterministic weighted automaton (see Proposition \ref{pro:dffa} below). However, $A$ can be accepted by a sequential weighted automaton ${\cal A}=(\{q\},\Sigma,\delta,\{q\},\{q\})$, where $\delta(q,a,p)=1$ if $p=q$ and $\delta(q,a,p)=0$ otherwise.

If $S$ is \emph{locally finite}, however, then the two definitions are equivalent in the sense that they accept the same class of former power series, where a semiring is locally finite if every sub-semiring generated by a finite set is also finite \cite{droste10}.
 \end{remark}
 Two weighted automata can be compared by a morphism.
 \begin{definition}\label{de:morphism}
A \emph{ homomorphism} (or \emph{ morphism}) between two weighted automata ${\cal A}=(Q,\Sigma,\delta,I,F)$ and ${\cal B}=(P,\Sigma,\eta,J,G)$ is a mapping $\varphi: Q\rw P$, satisfying the following conditions:
\begin{equation}\label{eq:morphism}
\mbox{$I(p)\leq J(\varphi(p))$, $F(p)\leq G(\varphi(p))$ and $\delta(p,\sigma,q)\leq \eta(\varphi(p),\sigma,\varphi(q))$},
\end{equation}
for any $p,q\in Q$ and $\sigma\in\Sigma$.

A morphism $\varphi$ is \emph{ surjective} if $\varphi: Q\rw P$ is onto and ${\cal B}=(P,\Sigma,\varphi(\delta),\varphi(I)$, $\varphi(F))$, where
\begin{eqnarray*}
\varphi(\delta)(p_1,\sigma,p_2) &=& \sum\{\delta(q_1,\sigma,q_2) | \varphi(p_1)=q_1, \varphi(p_2)=q_2\},\\
\varphi(I)(p) &=& \sum\{I(q)| \varphi(p)=q\},\\
\varphi(F)(p) &=& \sum\{F(q) | \varphi(p)=q\}.
\end{eqnarray*}
In this case, ${\cal B}$ is also called the \emph{morphic image} of ${\cal A}$.

If $\varphi:Q\rw P$ is one-to-one, then we call ${\cal A}$ a \emph{sub-automaton} of ${\cal B}$.

A morphism $\varphi$ is called a \emph{strong homomorphism} if
\begin{eqnarray*}
J(p) &=& \sum\{I(q) | \varphi(q)=p\},\\
G(\varphi(q)) &=& F(q),\\
\eta(\varphi(q),\sigma,p) &=& \sum\{\delta(q,\sigma,r) | \varphi(r)=p\}.
\end{eqnarray*}
In case $\varphi$ is an onto strong homomorphism,  we call ${\cal B}$ a \emph{quotient} of ${\cal A}$.

We say $\varphi$ is an \emph{isomorphism} if it is bijective and its inverse $\varphi^{-1}$ is also a morphism. In this case, we say ${\cal A}$ is isomorphic to ${\cal B}$.

\end{definition}

The behaviour of a weighted automaton is characterized by the formal power series it recognizes. To introduce this formal power series, we extend the weighted transition function $\delta: Q\times \Sigma \times Q \rightarrow S$ to a mapping $\ds:  Q\times \Sigma^\ast \times Q\rightarrow S$ as follows:

\begin{itemize}
\item [(i)] For all $p\in Q$,  set $\ds(q, \varepsilon, p)=1$ if $p=q$, and $\ds(q, \varepsilon, p)=0$ otherwise;
\item [(ii)] For all $\theta=\sigma_1\cdots \sigma_n\in\Sigma^{\ast}$, define
\begin{equation*}
\label{eq:ds}
\ds(q, \sigma_1\cdots \sigma_n, p)=\sum\{\delta(q,\sigma_1,q_1)\otimes\cdots\otimes\delta(q_{n-1},\sigma_n,p)| q_1,\cdots,q_{n-1}\in Q\}.
\end{equation*}
\end{itemize}
If ${\cal A}$ is deterministic, the extension $\ds$ of transition function $\delta$ is defined similar as in the classical case. It is easy to see that for any $\theta=\theta_1\theta_2\in \Sigma^{\ast}$ we have
\begin{equation}
\label{eq:theta1and2}
 \ds(q, \theta_1\theta_2, p)=\sum_{r\in Q}[\ds(q, \theta_1, r)\otimes\ds(r, \theta_2, p)].
\end{equation}

\begin{definition}\label{dfn:recognized-fps}
For a weighted automaton ${\cal A}=(Q,\Sigma,\delta,I,F)$, the formal power series  \emph{recognized} or \emph{accepted} by ${\cal A}$, written$|{\cal A}|: \sa \rightarrow S$, is defined as follows:
\begin{equation}
\label{eq:recognized-fps}
|{\cal A}|(\theta)=\sum\{I(p)\otimes\ds(p, \theta,q)\otimes F(q) | p, q\in Q\}. \hspace*{8mm} (\theta\in\sa)
\end{equation}
If ${\cal A}$ is deterministic, then the formal power series recognized by $\mathcal{A}$ is defined as
\begin{equation}
\label{eq:recognized-fps-det}
|{\cal A}|(\theta)=F(\ds(q_0,\theta)).  \hspace*{8mm} (\theta\in\sa)
\end{equation}
We say a formal power series $A\in \allseries$ is  a \emph{regular series} or an \emph{ $S$-regular language} on $\Sigma$ if it is recognized by a finite weighted automaton; and say $A$ is a \emph{DWA-regular series} or a \emph{DWA-regular language} on $\Sigma$ if it is recognized by a finite DWA.
\end{definition}
 
 It was proved by Sch\"{u}tzenberger \cite{schutzenberger61} that regular series are precisely the rational formal power series for all semirings. So we also say a regular series as a rational series in this paper.

\subsection{Weighted Contex-Free Grammar}
We next recall the concept of weighted context-free grammar and weighted context-free series.

A \emph{weighted context-free grammar} is in essence a classical context-free grammar together with a function mapping rules of the grammar to weights in a certain semiring (\cite{droste09,esik10}). Let $S$ be a semiring. A weighted context-free grammar (WCFG) is defined as a  tuple $G=(\Sigma,N,Z_0,S,P)$, where $\Sigma$ (the set of terminal symbols) and $N$ (the set of non-terminal symbols) are two finite sets that are disjoint, $Z_0$ (the start or initial symbol) is an element in $N$, $P$ is a mapping from $N \times (N \cup\Sigma)^{\ast}$ (the set of productions or rules) to $S$.

Similar to the classical case,  we can define the induction of a weighted context-free grammar $G$. Suppose $Z\overset{r}{\rightarrow}\gamma$ is a weighted production, and $\alpha,\beta$ are elements in  $(N\cup\Sigma)^{\ast}$. We say $\alpha\gamma\beta$ is a \emph{direct induction} of $\alpha Z\beta$ with weight $r$, denoted by $\alpha Z\beta\overset{r}{\Rightarrow}\alpha\gamma\beta$. For productions $\alpha_1,\cdots,\alpha_k$ in $(N\cup\Sigma)^{\ast}$, if $\alpha_1\overset{r_1}{\Rightarrow}\alpha_2, \cdots, \alpha_{k-1}\overset{r_{k-1}}{\Rightarrow}\alpha_k$, then we say $\alpha_k$ is an induction of $\alpha_1$ with weight $r=r_1\otimes\cdots\otimes r_{k-1}$, denoted by  $\alpha_1\overset{r}{\Rightarrow}_{\ast}\alpha_k$.
The formal power series $|G|$ generated by $G$ is defined as $|G|(w)=\sum\{r | Z_0\overset{r}{\Rightarrow}_{\ast} w\}$ ($w\in\sa$). A series $A$ is called \emph{ context-free} if there is a weighted context-free grammar $G$ such that $A=|G|$.

\subsection{Operations of Formal Power Series}
We recall several well-know operations of formal power series (cf. \cite{li05}). For two formal power series $A,B\in \allseries$, a value $r\in S$, and $w\in \Sigma^\ast$, we define
\begin{eqnarray}
(A\oplus B)(w) &=& A(w)\oplus B(w),\\
(AB)(w) &=& \sum\{A(w_1)\otimes B(w_2) | w_1w_2=w\},\\
(rA)(w) &=& r\otimes A(w),\\
(Ar)(w) &=& A(w)\otimes r,\\
A^{\ast}(w) &=& \sum\{A(w_1)\otimes\cdots\otimes A(w_n)| n\geq 0, w_1\cdots w_n=w\},\\
\label{eq:A^R}
A^R(w) &=& A(w^R),
\end{eqnarray}
where $w^R=\sigma_n\cdots\sigma_2\sigma_1$ if  $w=\sigma_1\sigma_2\cdots\sigma_n$.
We call $A\oplus B$ and $AB$ the \emph{sum} and, respectively, the \emph{concatenation} (or \emph{Cauchy product}) of $A$ and $B$, and call $rA$, $Ar$, $A^\ast$, and $A^R$ the \emph{left scalar product}, the \emph{right scalar product}, the \emph{Kleene closure}, and the \emph{reversal} of $A$, respectively.

Given a weighted automaton, ${\cal A}=(Q,\Sigma, \delta,I,F)$, three other operations can be defined for the formal power series recognized by $\mathcal{A}$. For any two states $p,q$ in $Q$,  and any $\theta\in \Sigma^\ast$, we define
\begin{eqnarray}
Past_{{\cal A}}(q)(\theta) &=& \sum\{I(p)\otimes \ds(p,\theta,q)| p\in Q\}, \\
Fut_{{\cal A}}(q)(\theta)   &=& \sum\{ \ds(q,\theta,p)\otimes F(p)| p\in Q\}, \\
Trans_{{\cal A}}(p,q)(\theta) &=& \ds(p,\theta,q).
\end{eqnarray}
The following result holds.
\begin{proposition}\label{pro:factor-automata}
Suppose ${\cal A}=(Q,\Sigma, \delta,I,F)$ is a weighted automaton. For any $q\in Q$, we have
\begin{equation}
Past_{{\cal A}}(q)Fut_{{\cal A}}(q)\leq |{\cal A}|.
\end{equation}
\end{proposition}

\begin{proof}
For any $\theta\in \sa$, we have
\begin{eqnarray*}
&& (Past_{{\cal A}}(q)Fut_{{\cal A}}(q))(\theta) \\
 &=& \sum_{uv=\theta}Past_{{\cal A}}(q)(u)\otimes Fut_{{\cal A}}(q)(v)\\
                                                                   &=& \sum_{uv=\theta}\sum_{q_0\in Q}I(q_0)\otimes \ds(q_0,u,q)\otimes \sum_{p\in Q}\ds(q,v,p)\otimes F(p)\\
                                                                   &=& \sum_{uv=\theta}\sum_{q_0,p\in Q}I(q_0)\otimes\ds(q_0,u,q)\otimes\ds(q,v,p)\otimes F(p)\\
                                                                   &\leq& \sum_{uv=\theta}\sum_{q_0,p,q\in Q}I(q_0)\otimes\ds(q_0,u,q)\otimes\ds(q,v,p)\otimes F(p)\\
                                                                   &=& \sum_{uv=\theta}I(q_0)\otimes\ds(q_0,uv,p)\otimes F(p)\\
                                                                   &=& |{\cal A}|(\theta).
\end{eqnarray*}
Hence, $Past_{{\cal A}}(q)Fut_{{\cal A}}(q)\leq |{\cal A}|$.
\end{proof}

\section{Quotients of Formal Power Series}
In this section, we first introduce quotients of formal power series and then, upon this operation, introduce for each formal power series $A$ a canonical weighted automaton that recognizes $A$. Properties of the quotient operation is also studied.

\subsection{Quotients and Minimal Weighted Automata}

Let $A: \sa\rw S$ be a formal power series, $u\in\sa$ a word. The \emph{left quotient} of $A$ by $u$, written $u^{-1}A$, is the formal power series $u^{-1}A:\sa\rw S$ defined as:
\begin{equation}
\label{eq: left-quotient}
u^{-1}A (v)=A(uv) \hspace{8mm} (v\in\sa).
\end{equation}

Dually, the  \emph{right quotient} of $A$ by $u$, written $Au^{-1}$, is the formal power series $Au^{-1}: \sa\rw S$ defined as:
\begin{equation}
\label{eq: left-quotient}
Au^{-1}(v)=A(vu) \hspace{8mm} (v\in\sa).
\end{equation}
for any $v\in\sa$.




The left quotient operation introduces an equivalent relation on $\sa$.

\begin{definition}\label{dfn:non-distinguishable-relation}
Suppose $A\in \allseries$. We say two words $u_1$ and $u_2$ are \emph{non-distinguishable} in $A$, written $u_1 \equiv_A u_2$, if $u_1^{-1} A=u_2^{-1}A$, i.e., if $A(u_1u)=A(u_2 u)$ for all $u\in\sa$.
\end{definition}
It is straightforward to show that $\equiv_A$ is an equivalent relation on $\sa$. For each word $u\in\sa$, we write $[u]_A$ for the equivalent class of $\equiv_A$ that contains $u$.
\begin{lemma}\label{le:mini DWA}
Suppose $A \in \allseries$.  The mapping defined by $[u]\mapsto u^{-1}A$ $(u\in \sa)$ is a bijection from the set of equivalent classes of $\equiv_A$ to the set of left quotients of $A$.
\end{lemma}

We define a deterministic weighted automaton ${\cal M}_A=(Q_A,\Sigma,\delta_A,I_A,F_A)$ as follows:
\begin{itemize}
\item $Q_A=\sa/\equiv_A$ is the quotient of $\sa$ modulo $\equiv_A$;
\item $\delta_A: Q_A\times \Sigma\rw Q_A$ is defined as
\begin{equation}
\label{eq: delta-minimal-dwa}
\delta_A([\theta],\sigma)=[\theta\sigma], \hspace*{8mm} (\sigma\in\sa)
\end{equation}

\item $I_A$ is the singleton state $[\varepsilon]$ in $Q_A$;

\item $F_A: Q_A\rw S$ is defined by $F_A([\theta])=A(\theta)$ for $\theta\in \sa$.
\end{itemize}

\begin{proposition}\label{pro:minimal DWA}
Suppose $A\in\allseries$. Then ${\cal M}_A$ is the minimal DWA that recognizes $A$.

\end{proposition}

It is easy to see that $\delta_A$ is well-defined and, hence, ${\cal M}_A=(Q_A,\Sigma,\delta_A,I_A,F_A)$ is a DWA. It is straightforward to show that ${\cal M}_A$ recognizes $A$. In other words, there is a DWA ${\cal A}$ that accepts $A$ for any formal power series $A$. In general, ${\cal M}_A$ is not finite; but, when it is finite, ${\cal M}_A$ is the minimal DWA that recognizes $A$.

By Lemma \ref{le:mini DWA}, an equivalent minimal DWA
\begin{equation}
\label{eq:M-A'}
{\cal M}_A'=(Q_A',\Sigma,\delta_A',I_A',F_A')
\end{equation}
that recognizes $A$ can be constructed by using the left quotients of $A$ as follows

\begin{itemize}
\item $Q_A'=\{u^{-1}A | u\in\sa\}$, the set of all left quotients of $A$ by a word;
\item  $\delta_A': Q_A'\times \Sigma\rw Q_A'$ defined by
\begin{equation}
\label{eq:delta'}
\delta_A'(u^{-1}A,\sigma)=(u\sigma)^{-1}A
\end{equation}

\item $I_A'=A=\varepsilon^{-1}A$;

\item $F_A':Q_A'\rw S$ is defined by $F_A'(u^{-1}A)=A(u)$ for $\theta\in \sa$.
\end{itemize}

The following proposition presents a characterization of formal power series that can be recognized by a finite DWA.

\begin{proposition}\label{pro:dffa}
Suppose $A\in \allseries$. Given $r\in S$, we write $Im(A)=\{A(\theta) | \theta\in\sa\}$, $A_r=\{ \theta\in\sa | A(\theta)\geq r\}$, and $A_{[r]}=\{ \theta\in\sa | A(\theta)=r\}$. Then the following statements are equivalent.
\begin{itemize}
\item [(i)] $A$ can be recognized by a finite DWA.

\item [(ii)] $Im(A)$ is finite and each $A_r$ is a regular language for $r\in Im(A)$.

\item [(iii)] $Im(A)$ is finite and each $A_{[r]}$ is a regular language for $r\in Im(A)$.

\item [(iv)] There exist  $r_1,\cdots,r_k\in S\setminus\{0\}$ and $k$ regular languages $L_1,\cdots,L_k$ over $\Sigma$, which are pairwise disjoint, such that $A=\sum_{i=1}^k r_i L_i$.

\item [(v)] $\equiv_A$ has finite index, i.e.,  $Q_A$ is finite.
\end{itemize}

\end{proposition}
\begin{proof}
The  proof is similar to that given for lattices in  \cite{li05,li11}.
\end{proof}

But \emph{when can a regular series be recognized by a finite DWA}? This is closely related with the structure of semiring $S$. It can be shown that (cf.\cite{li05}), for any regular series $A$, $A$ can be recognized by a finite DWA iff the monoid $(S,\oplus,0)$ and $(S,\otimes,1)$ are both locally finite, where a monoid $(M,\times,1)$ is locally finite if every submonoid generated by a finite subset of $M$ is also finite. In particular, if $S$ is finite, regular series and  DWA-regular series are the same.

\subsection{Properties of Quotients}

The left quotient of a series $A$ by a word $u$ can be regarded as a left action of $\sa$ on $\allseries$, i.e., a mapping $\sa\times \allseries\rw \allseries$. The action satisfies the following conditions.

\begin{proposition}\label{pro:quotient-word}
Let $S$ be a semiring. Suppose $A\in \allseries$, $\varepsilon$ is the empty word in $\sa$, $u,v$ are words in $\sa$, and $k$ is a value in $S$. Then we have
\begin{itemize}

\item[(i)] $(uv)^{-1}A=v^{-1}(u^{-1}A)$, $A(uv)^{-1}=(Av^{-1})u^{-1}$, $\varepsilon^{-1}A=A=A\varepsilon^{-1}$;

\item[(ii)]$u^{-1}(A\oplus B)=u^{-1}A\oplus u^{-1}B$, $(A\oplus B)u^{-1}=Au^{-1}\oplus Bu^{-1}$;

\item [(iii)] $u^{-1}(kA)=k(u^{-1}A)$, $u^{-1}(Ak)=(u^{-1}A)k$, $(kA)u^{-1}=k(Au^{-1})$, $(Ak)u^{-1}$ $=(Au^{-1})k$.
\end{itemize}
\end{proposition}
\begin{proof}
 Straightforward.
\end{proof}

When the semiring $S$ is complete, the left quotient operation can be extended from words to series in a natural way. Let $A, X$ be two formal power series in $\allseries$. We define the \emph{left quotient} of $A$ by $X$, denoted by $X^{-1}A$, as
\begin{equation}
\label{eq:X-1A}
X^{-1}A(v)=\sum_{u\in\sa}X(u) (u^{-1}A)(v)=\sum_{u\in\sa}X(u) A(uv). \hspace*{8mm} (v\in\sa)
\end{equation}

Similarly, we can define the \emph{right quotient} of $A$ by $Y$ for any $Y\in \allseries$, denoted by $AY^{-1}$ as
\begin{equation}
\label{eq:AY-1}
AY^{-1}(v)=\sum_{u\in \sa}A(vu) Y(u). \hspace*{8mm} (v\in\sa)
\end{equation}

It is easy to see that when the formal power series $X$ is a word $u$, then $X^{-1}A=u^{-1}A$ and $AX^{-1}=Au^{-1}$. We summarize some algebraic properties of the quotient operation.

\begin{proposition}\label{pro:quotient-fps}
Let $S$ be a complete semiring. Suppose $A,X_1,X_2 \in \allseries$. Then

\begin{itemize}
\item
[(i)] $(X_1\oplus X_2)^{-1}A=X_1^{-1}A\oplus X_2^{-1}A$, $A(Y_1\oplus Y_2)^{-1}=AY_1^{-1}\oplus AY_2^{-1}$;

\item
[(ii)] $X^{-1}(A_1\oplus A_2)=X^{-1}A_1\oplus X^{-1}A_2$, $(A_1\oplus A_2)Y^{-1}=A_1Y^{-1}\oplus A_2Y^{-1}$;

\item
[(iii)] $(X^{\ast})^{-1}A=A\oplus (X^{\ast})^{-1}(X^{-1}A)$, $A(Y^{\ast})^{-1}=A\oplus (AY^{-1})(Y^{\ast})^{-1}$;

\item[(iv)]$X^{-1}(AB)=(X^{-1}A)B\oplus (A^{-1}X)^{-1}B$, $(AB)Y^{-1}=A(BY^{-1})\oplus A(YB^{-1})^{-1}$;

\item  [(v)] $(rX)^{-1}A=r(X^{-1}A)$, $X^{-1}(Ar)=(X^{-1}A)r$, $(rA)Y^{-1}=r(AY^{-1})$, $A(Yr)^{-1}=(AY^{-1})r$;


\item [(vi)] If $S$ is commutative, then $(X_1X_2)^{-1}A=X_2^{-1}(X_1^{-1}A)$,  $A(Y_1Y_2)^{-1} = (AY_2^{-1})Y_1^{-1}$.


\end{itemize}
\end{proposition}

If the semiring $S$ is commutative, the following lemma shows that the right quotient is dual to the left quotient, where the reversal operation $X^R$ is defined (see Eq.~\ref{eq:A^R}) as $X^R(\sigma_1\sigma_2\cdots\sigma_n)=X(\sigma_n\cdots\sigma_2\sigma_1)$ for any $w=\sigma_1\sigma_2\cdots\sigma_n$.
\begin{lemma}
\label{lemma:right-quotient}
Suppose $S$ is a complete semiring. For $A,X, Y\in\allseries$, if $Y(u)A(vu)=A(vu)Y(u)$, $A(uv)X(u)=X(u)A(uv)$ for any $u,v\in\sa$, then
\begin{eqnarray}
 AY^{-1} &=& ((Y^R)^{-1}A^R)^R\\
 X^{-1}A &=& (A^R(X^{R})^{-1})^R.
 \end{eqnarray}
In particular, if $A$ is a classical language, i.e., $Im(A)\subseteq \{0,1\}$ or if $S$ is commutative, the above equalities hold.
\end{lemma}
\begin{proof}
Straightforward.
\end{proof}
By the above result, if the semiring $S$ is commutative, then the right quotient is dual to the left quotient by the reversal operation. Properties of the right quotient can be dually obtained in this case.

\subsection{Closure Properties of Former Power Series under Quotient}
In this subsection, we study the language properties of the quotient of series. We first recall a preliminary result which will be used in the proof.

\begin{lemma} [cf. \cite{droste09}]\label{le:proper}
Suppose $Y$ is a proper and regular series in $\allseries$. Then there exists a weighted automaton ${\cal A}=(Q,\Sigma,\delta,q_0,\{q_f\})$ that recognizes  $Y$ such that
$q_0\not= q_f$ and $\delta(q,\sigma,q_0)=0$ for any $q\in Q$.

\end{lemma}

\begin{proposition}\label{pro:quotient-closeness}
Let $S$ be a complete semiring. Suppose $A,X,Y$ are former power series in $\allseries$. Then
\begin{itemize}
\item [(i)] If $A$ is regular, then $X^{-1}A$ and $AY^{-1}$ are regular.

\item [(ii)]  If $A$ is DWA-regular, then $X^{-1}A$ and $AY^{-1}$ are DWA-regular.

\item [(iii)] For a commutative semiring $S$, if $A$ is context-free, $X$ and $Y$ are regular, then $X^{-1}A$ and $AY^{-1}$ are context-free.

\item [(iv)] If $A$ is context-free, $X$ and $Y$ are DWA-regular, then $X^{-1}A$ and $AY^{-1}$ are context-free.
\end{itemize}
\end{proposition}

\begin{proof}
See Appendix A.
\end{proof}

Using the quotient of a series $A$ by a series, we have a related deterministic weighted automaton
\begin{equation}
\label{eq:B_A}
{\cal B}_A=(\{X^{-1}A | X\in \allseries\}, \Sigma, \delta,A,F),
\end{equation}
where $\delta(X^{-1}A,\sigma)=(X\sigma)^{-1}A$, and $F(X^{-1}A)=\sum_{u\in\sa}X(u)A(uv)$.
Because the accessible part of ${\cal B}_A$ is the minimal DWA ${\cal M}_A'$ (cf. (\ref{eq:M-A'})), hence $|{\cal B}_A|=A$.

Recall in the classical case, if $A$ is a regular language over $\Sigma$, then $\{X^{-1}A | X\subseteq \sa\}$ (the set of all left quotients of $A$ by languages) is a finite set. This does not hold in general. In fact, the finiteness of $Q=\{X^{-1}A | X\subseteq \allseries\}$ heavily depends on the finiteness of the semiring $S$.

\begin{proposition}\label{pro:quotient automaton}
Let $S$ be a semiring. Then $S$ is finite if and only if ${\cal B}_A$ is finite for any DWA-regular series $A$ in $\allseries$, where  $\mathcal{B}_A$ is the weighted automaton that recognizes $A$ defined in (\ref{eq:B_A}).
\end{proposition}

\begin{proof}
Suppose ${\cal B}_A$ is finite for any DWA-regular series $A$. In particular, ${\cal B}_A$ is finite for $A=\sa$. Take $a\in \Sigma$, for any $r\in S$, let $X_r=ra$ (which is a series over $\Sigma$). By a simple calculation, the left quotient $X_r^{-1}A$ is $\sum_{w\in\sa}rw$, i.e., $X_r^{-1}A=\overline{r}$, where $\overline{r}(w)=r$ for any $w\in\sa$. Then, as a subset of $\{X^{-1}A | X\in\allseries\}$, the set $\{\overline{r} | r\in S\}$ is finite. Hence, $S$ is a finite set.

Conversely, suppose $S=\{r_1,\cdots,r_k\}$ is finite. For any $X\in\allseries$, let $X_i=\{w\in\sa | X(w)=r_i\}$. Then $X=\sum_{i=1}^k r_iX_i$. Since $A$ is regular, there exist finite regular languages $L_1,\cdots,L_k$ such that $A=\sum_{i=1}^k r_iL_i$. By Proposition \ref{pro:quotient-fps}, it follows that $X^{-1}A=\sum_{i=1}^k\sum_{j=1}^k(r_i\otimes r_j)X_i^{-1}L_j$. Since $L_j$ is regular, the set $\{X_i^{-1}L_j |X\subseteq \sa\}$ is finite. It follows that the set $\{X^{-1}A |X\in \allseries\}$ is also finite.
\end{proof}

\par

In the remainder of this paper, we will consider the residual operations on formal power series.
\section{Residuals and Factorizations}

In this section, we study the residuals and factorizations of formal power series. For $A,X\in \allseries$, the \emph{left residual} of $A$ by $X$, written as $X\backslash A$, is defined as the largest $Y\in\allseries$ such that $XY\leq A$.  Similarly, we define the \emph{right residual} of $A$ by $X$, written as $A/X$, as the largest $Y\in\allseries$ such that $YX\leq A$. We note that  $X\backslash A$ and $A/X$ need not exist for all $A,X$.

\subsection{Residuals of Formal Power Series}
The following proposition shows that residuals always exist when $S$ is a complete c-semiring. We recall that each c-semiring is commutative.

\begin{proposition}
Suppose $S$ is a complete c-semiring. For any $A,X,Y\in\allseries$, $X\backslash A$ and $A/Y$ do exist and we have
\begin{eqnarray}
 X\backslash A &=& \sum\{Z | XZ\leq A\},\\
 A/Y &=& \sum\{Z | ZY\leq A\}.
\end{eqnarray}
\end{proposition}
\begin{proof}
Because $S$ is a complete c-semiring, we know $\sum_{i\in I}Z_i$ exists for any subset $\{Z_i | i\in I\}$ of $\allseries$. Moreover, the concatenation operation satisfies the following conditions:
\begin{eqnarray}
X(\sum_{i\in I}Y_i) &=& \sum_{i\in I} XY_i,\\
(\sum_{i\in I}Y_i)X &=& \sum_{i\in I} Y_iX.
\end{eqnarray}
That is to say, $\allseries$ is a quantale under the operations $\sum$ and concatenation. It is then easy to see that  $X (\sum\{Y | XY\leq A\}) \leq A$ and hence  $\sum\{Y | XY\leq A\}$ is the largest $Z$ such that $XZ\leq A$, i.e., $A\backslash Y= \sum\{Y | XY\leq A\}$. Similarly, we have $A/Y=\sum\{X | XY\leq A\}$.
\end{proof}


The following proposition shows that left residual and right residual are dual to each other.
\begin{proposition}
\label{prop:dual-residuals}
Let $S$ be a complete c-semiring.  Suppose $A,X,Y\in\allseries$. Then we have
$(X\backslash A)^R=A^R/X^R$ and $(A/Y)^R$ $=Y^R\backslash A^R$.
\end{proposition}
\begin{proof}
Straightforward.
\end{proof}
By the above result, in the following, we often consider only one (either left or right) residual. Dual results can be applied to the other residual.

Left residual has close connection with left quotient.
\begin{proposition}\label{pro:quotient-word}
Let $S$ be a complete c-semiring. Suppose $A$ is a formal power series in $\allseries$, $u$ is a word in $\sa$. Then $u^{-1}A$ is the largest language $Y$ such that $uY\leq A$, i.e., $u\backslash A = u^{-1}A$.
\end{proposition}
\begin{proof}
Note that $(u(u^{-1}A))(\theta)=u^{-1}A(v)=A(uv)$ if $\theta=uv$ and $0$ otherwise. It follows that $u(u^{-1}A)\leq A$. If $uY\leq A$, then $Y(v)=(uY)(uv)\leq A(uv)=(u^{-1}A)(v)$ for any $v\in\sa$. Hence, $Y\leq u^{-1}A$. Therefore $u^{-1}A$ is the largest  $Y$ such that $uY\leq A$.
\end{proof}
In the classical case, given two languages $A,X$ over $\Sigma$, we always have
\begin{eqnarray}
X\backslash A &=& \bigcap \{u^{-1}A | u\in X\}\\
A/X &=& \bigcap\{Au^{-1} | u\in X\}.
\end{eqnarray}
In the weighted case,  residuals are usually not the intersection of a set of quotients by word. But we can still represent residuals in terms of quotients by word. Before give the details, we summarize some basic algebraic properties of the residuals.

\begin{proposition}\label{pro:residual quotient}
Let $S$ be a complete c-semiring. For $A,X,Y\in\allseries$, we have
\begin{itemize}
\item
[(i)] $X(X\backslash A)\leq A$, $(A/Y)Y\leq A$.

\item [(ii)] The operations $cl,cr: \allseries\rw \allseries$ defined by $cl(X)=A/(X\backslash A)$ and $cr(Y)=(A/Y)\backslash A$ are two closure operators, where $\tau: \allseries \rw \allseries$ is a closure operator if $X\leq \tau(X)$, $\tau(X_1\vee X_2)=\tau(X_1)\vee \tau(X_2)$ and $\tau(\tau(X))=\tau(X)$.

\item [(iii)] $(XY)\backslash A=Y\backslash(X\backslash A)$, $A/(YX)=(A/X)/Y$.

\item [(iv)] $(X_1\oplus X_2)\backslash A=X_1\backslash A \wedge X_2\backslash A$, $A/(Y_1\oplus Y_2)=A/Y_1\wedge A/Y_2$.

\item [(v)] $X\backslash (A\wedge B)=(X\backslash A)\wedge (Y\backslash B)$.

\item [(vi)] $(A/(X\backslash A))\backslash A=X\backslash A$, $A/((A/Y)\backslash A)=A/Y$.

\end{itemize}
\end{proposition}
\begin{proof}
We give the proof of (vi), the others are straightforward. By (ii), $X\leq A/(X\backslash A)$, then it follows that $(A/(X\backslash A))\backslash A\leq X\backslash A$. On the other hand, since $A/(X\backslash A) (X\backslash A)\leq A$, it follows that $X\backslash A\leq (A/(X\backslash A))\backslash A$. Hence, $(A/(X\backslash A))\backslash A=X\backslash A$. Similarly, we have $A/((A/Y)\backslash A)=A/Y$.
\end{proof}

We next give the characterization of residuals in terms of quotients by word, where $\rw$ is the residual operation in the quantale $S$ defined (cf. (\ref{a-rw-b})).

\begin{proposition}\label{pro:characterization of residual quotient}
Let $S$ be a complete c-semiring.
Suppose $X,Y,A\in \allseries$. Then, for any $v\in \sa$, we have
\begin{eqnarray}
\label{eq:X|A}
X\backslash A(v) &=& \bigwedge\{X(u)\rw A(uv) | u\in \sa\}\\
\label{eq:A|Y}
A/Y(v)                  &=& \bigwedge\{Y(u)\rw A(vu) | u\in \sa\}.
\end{eqnarray}
\end{proposition}

\begin{proof}
By Proposition~\ref{prop:dual-residuals}, we need only prove Eq.(\ref{eq:X|A}).

For $v\in\sa$, let $\overline{Y}(v)=\bigwedge\{X(u)\rw A(uv) | u\in \sa\}$. Then $\overline{Y}$ is a series. We show $\overline{Y}=X\backslash A$, i.e., $\overline{Y}$ is the largest series $Y$ such that $XY\leq A$.

First,
\begin{eqnarray*}
X\overline{Y}(w) &=& \sum_{uv=w}X(u)\overline{Y}(v) \\
&\leq& \sum_{uv=w}X(u)(X(u)\rw A(uv)) \\
&\leq& \sum_{uv=w}A(uv)=A(w).
\end{eqnarray*}

Second, if $XY\leq A$, then, for any $uv=w$, $X(u)Y(v)\leq A(uv)$. It follows that $Y(v)\leq X(u)\rw A(uv)$ for any $u\in\sa$, thus, $Y(v)\leq \bigwedge\{X(u)\rw A(uv) | u\in\sa\}=\overline{Y}(v)$, i.e., $Y\leq \overline{Y}$. This shows that $X\backslash A=\overline{Y}$.
\end{proof}
The following proposition shows that DWA-regular series are closed under residual operations.
\begin{proposition}\label{pro:implication-language}
Suppose $S$ is a complete c-semiring,
 and $A,X,Y\in \allseries$. If $A$ is DWA-regular, then so are $X\backslash A$ and $A/Y$.
\end{proposition}
\begin{proof}
By Proposition~\ref{prop:dual-residuals}, we need only consider right residuals. Suppose that ${\cal A}=(Q,\Sigma,\delta,q_0,F)$ is a DWA accepting $A$. Then we have
$|{\cal A}|(\theta)=F(\ds(q_0,\theta))$ for any $\theta\in\sa$.

Define another weighted automaton, ${\cal A}^Y=(Q,\Sigma,\delta,q_0,F^Y)$, where
$$F^Y(q)=\bigwedge\{Y(u)\rw F(\ds(q,u)) | u\in \sa\}.$$

Then
\begin{eqnarray*}
|{\cal A}^Y|(\theta) &=& F^Y(\ds(q_0,\theta))\\
                             &=& \bigwedge\{Y(u)\rw F(\ds(\ds(q_0,\theta),u) | u \in \sa\}\\
                             &=& \bigwedge\{Y(u)\rw F(\ds(q_0,\theta u) | u \in \sa\} \\
                             &=& \bigwedge\{Y(u)\rw A(\theta u) | u \in \sa\} \\
                             &=& A/Y(\theta).
\end{eqnarray*}
Hence, $A/Y$ is DWA-regular.
\end{proof}

\subsection{Factorizations of Formal Power Series}
In formal language theory, factorization is an important notion that is closely related to quotients and residuals \cite{lombardy08}. This notion can be generalized to formal power series straightforwardly.
\begin{definition} \label{de:factor}
For any $X,Y\in \allseries$, if $XY\leq A$, then we call $(X,Y)$ a sub-factorization of $A$. Furthermore, if the sub-factorization $(X,Y)$ is maximal,
i.e., if $X\leq X^{\prime}$, $Y\leq Y^{\prime}$ and $X^{\prime}Y^{\prime}\leq A$, then $X=X^{\prime}$ and $Y=Y^{\prime}$. In this case we call $(X,Y)$ a factorization
of $A$, and write $R_A$ for the set of all factorizations of $A$.
\end{definition}

The relationship between residuals and factorizations of a series is as follows.

\begin{proposition}\label{pro:factor-language}
Let $S$ be a complete c-semiring. For $A,X,Y\in \allseries$, we have
\begin{itemize}
\item
[(i)] $(X,Y)\in R_A$ if and only if $X=A/Y$ and $Y=X\backslash A$.

\item [(ii)] If $WZ\leq A$, then there exists  $(X,Y)\in R_A$ such that $W\leq X$ and $Z\leq Y$.
\end{itemize}
\end{proposition}

\begin{proof}
(i) Suppose $(X,Y)\in R_A$ we show $X=A/Y$ and $Y=X\backslash A$. For any $u,v\in\sa$, we note that $X(u)\otimes Y(v)\leq A(uv)$ if and only if $Y(v)\leq X(u)\rw A(uv)$. Hence, $Y(v)\leq \bigwedge_{u\in \sa}X(u)\rw A(uv)=X\backslash A(v)$. This shows that $Y\leq X\backslash A$. Conversely, by $X(X\backslash A)(\theta)=\bigvee_{uv=\theta}X(u)\otimes X\backslash A(v)\leq A(\theta)$ and the maximality of the factorization, we know $X\backslash A\leq Y$. Hence, $Y=X\backslash A$. Similarly, we can show $X=A/Y$.

On the other hand, suppose $X=A/Y$ and $Y=X \backslash A$. By definition, we know $Y$ is the largest $Z$ such that $XZ\leq A$, and $X$ is the largest $W$ such that  $WY\leq A$. This shows that $(X,Y)$ is a factorization of $A$.

(ii)  Let $X=A/Z$, $Y=(A/Z)\backslash A$. It is clear that $W\leq X$. By Prop.~\ref{pro:residual quotient} (ii) and (vi), we know $Z\leq (A/Z)\backslash A$ and $X=A/ [(A/Z)\backslash A] = A/Y$. Therefore $(X,Y)\in R_A$ and $W\leq X$ and $Z\leq Y$.
\end{proof}

By the above proposition, a factorization of a formal power series $A$ is just a pair $(X,Y)$ such that $X$ is the right residual of $A$ by $Y$ and $X$ is the left residual of $A$ by $X$. In this case, we also call $X$ the \emph{ left factor} of $A$ by $Y$ and $Y$ the \emph{ right factor} of $A$ by $X$, respectively.  Since $\varepsilon A=A\varepsilon=A$, by Proposition~\ref{pro:factor-language} (ii) $A$ itself is both a left factor and a right factor. We denote the corresponding right factor and left factor by $X_s$ and
$Y_e$, respectively, where $X_s(\varepsilon)=Y_e(\varepsilon)=1$, and call $(X_s,A)$ and $(A,Y_e)$ the \emph{initial} and \emph{final} factorization, respectively.

We write $lR(A)$ ($rR(A)$) for the set of left (right) residuals of $A$, i.e.,
\begin{eqnarray*}
lR(A) &=& \{X\backslash A | X\in \allseries\},\\
rR(A) &=& \{A/Y | Y\in \allseries\}.
\end{eqnarray*}
Let $\varphi:lR(A)\rw rR(A)$ be the mapping defined as $\varphi(X\backslash A)=A/(X\backslash A)$. By Proposition \ref{pro:residual quotient}, it is easy to see that $\varphi$ is a bijection.
Moreover, we have
$$R_A=\{(X,\varphi(X)) | X\in lR(A)\}=\{(\varphi^{-1}(Y),Y) | Y\in rR(A)\}.$$

\section{The Universal Weighted Automaton}
In this section, we use the residuals of a formal power series $A$ to construct a weighted automaton which recognizes $A$. To this end, we introduce the notion of the inclusion degree.

\begin{definition}\label{def:inclusion regree}
Suppose $S$ is a complete c-semiring. For two $S$-subsets $f,g:U \rw S$, the \emph{inclusion degree} of $f$ into $g$, denoted by
$f\rw_{incl} g\in S$, is defined as
\begin{equation}
f\rw_{incl} g=\bigwedge\{f(u)\rw g(u) | u\in U\}.
\end{equation}
\end{definition}

The following lemma summarizes several useful properties of the inclusion degree operator.

\begin{lemma}\label{le:inclusion}
Suppose $S$ is a complete c-semiring. For  $X,X^\prime\in \allseries$, any $c\in S$, and any $w\in \sa$, we have
\begin{itemize}
\item
[(1)]  $c\leq X\rw_{incl} X^{\prime}$ iff $cX\leq X^{\prime}$.

\item [(2)] If $(X,Y),(X^{\prime}, Y^{\prime})\in R_A$, then $X\rw_{incl}  X^{\prime}=Y^{\prime}\rw_{incl}  Y$.

\item [(3)] $X\rw_{incl}  X^{\prime}=Xw\rw_{incl}  X^{\prime}w=wX\rw_{incl}  wX^{\prime}$.
\end{itemize}
\end{lemma}

\begin{definition}\label{def:universal automaton}
Suppose $S$ is a complete c-semiring. For $A\in\allseries$, the \emph{universal weighted automaton} of $A$, denoted by ${\cal U}_A$, is a weighted automaton
$(R_A,\Sigma, \eta_A$, $J_A, G_A)$, where
\begin{eqnarray}
J_A(X,Y) &=& X(\varepsilon),\\
G_A(X,Y) &=& Y(\varepsilon),\\
\eta_A((X,Y),\sigma,(X^{\prime}, Y^{\prime})) &=& X\sigma Y^{\prime}\rw_{incl}  A,
\end{eqnarray}
for any $(X,Y),(X^{\prime}, Y^{\prime})\in R_A$, $\sigma\in \Sigma$.
\end{definition}

\begin{proposition}\label{pro:ufa1} Suppose $S$ is a complete c-semiring. For $A\in \allseries$, and $(X,Y), (X',Y')\in R_A$ and $\sigma\in \Sigma$, we have
\begin{eqnarray}
 X(\varepsilon) &=&Y\rw_{incl}  A,\\
 Y(\varepsilon) &=& X\rw_{incl}  A,\\
X\sigma Y^{\prime}\rw_{incl}  A &=& X\sigma\rw_{incl} X^{\prime}=\sigma Y^{\prime}\rw_{incl}  Y.
\end{eqnarray}
\end{proposition}
\begin{proof}
Since $(X,Y)\in R_A$,
\begin{eqnarray*}
X(\varepsilon)=A/Y(\varepsilon) &=& \bigwedge\{Y(v)\rw A(\varepsilon v) | v\in \sa\}\\
&=& \bigwedge\{Y(v)\rw A(v) | v\in \sa\}=Y\rw_{incl}  A.
\end{eqnarray*}
Similarly, we can prove the case for $G_A(X,Y)$.

As for  $\eta_A$, it is sufficient to show that $c\leq
X\sigma Y^{\prime}\rw_{incl}  A$ iff $c\leq X\sigma\rw_{incl}  X^{\prime}$ for any $c\in S$. This is because, $c\leq X\sigma Y^{\prime}\rw_{incl}  A$ iff
$cX\sigma Y^{\prime}\leq A$, iff $(cX\sigma) Y^{\prime}\leq A$, iff
$cX\sigma\leq X'$, iff $c\leq X\sigma\rw_{incl}  X^{\prime}$. Hence,
$X\sigma Y^{\prime}\rw_{incl}  A=X\sigma\rw_{incl}  X^{\prime}$.

Similarly, we have $X\sigma Y^{\prime}\rw_{incl}  A=\sigma Y^{\prime}\rw_{incl}  Y$.
\end{proof}
 The extension of $\eta$ has the following form.

\begin{proposition}\label{pro:extension}
Let $S$ be a complete c-semiring and suppose $A\in\allseries$.
For $(X,Y),(X^{\prime}$, $Y^{\prime})\in R_A$, and any $w\in \Sigma^{+}$, we have
\begin{equation}\label{eq:extension of eta}
\eta_A^{\ast}((X,Y),w,(X^{\prime}, Y^{\prime}))= XwY^{\prime}\rw_{incl}  A=Xw\rw_{incl} X^{\prime}=wY^{\prime}\rw_{incl}  Y.
\end{equation}
\end{proposition}

\begin{proof} First, for any $c\in S$, $c\leq XwY^{\prime}\rw_{incl}  A$ iff $cXwY^{\prime}\leq A$, iff $(cXw)Y^{\prime}\leq A$, iff $cXw\leq X^{\prime}$, iff $c\leq Xw\rw_{incl}  X^{\prime}$.
Hence, $XwY'\rw_{incl}  A=Xw\rw_{incl}  X^{\prime}$.

Similarly, we have $XwY^{\prime}\rw_{incl}  A=wY^{\prime}\rw_{incl}  Y$.

We shall show $\eta_A^{\ast}((X,Y),w,(X^{\prime}, Y^{\prime}))=
XwY^{\prime}\rw_{incl}  A$ by induction on the length $|w|$ for $w\in \Sigma^{+}$.

If $|w|=1$, this is just the definition of $\eta_A$.

Given $w\in \Sigma^{+}$ and $\sigma\in\Sigma$, we show that  (\ref{eq:extension of eta}) holds for $\sigma w$ if (\ref{eq:extension of eta}) holds for $w$.
\begin{eqnarray*}
&&\eta_A^{\ast}((X,Y),\sigma w,(X^{\prime}, Y^{\prime}))\\
 &=& \bigvee_{(X'',Y'')\in R_A}\eta_A((X,Y),\sigma,(X'', Y''))\otimes \eta_A^{\ast}((X'',Y''),w,(X^{\prime}, Y^{\prime}))\\
 &=& \bigvee_{(X'',Y'')\in R_A}(X\sigma Y''\rw_{incl}  A)\otimes (X''wY'\rw_{incl}  A)\\
  &=& \bigvee_{(X'',Y'')\in R_A}(X\sigma\rw_{incl}  X'')\otimes (X''w\rw_{incl}  X')\\
   &=& \bigvee_{(X'',Y'')\in R_A}(X\sigma w\rw_{incl}  X''w)\otimes (X''w\rw_{incl}  X')\\
    &\leq& \bigwedge_{\theta\in\sa}(X\sigma w(\theta)\rw  X''w(\theta))\otimes(X''w(\theta)\rw  X'(\theta)) \\
    &\leq&  \bigwedge_{\theta\in\sa}(X\sigma w(\theta)\rw  X'(\theta))\\
    &=& X\sigma w\rw_{incl}  X'.
\end{eqnarray*}
Conversely,
\begin{eqnarray*}
&& c\leq X\sigma w\rw_{incl} X' = X\sigma w Y'\rw_{incl}  A \\
& \Longleftrightarrow& cX\sigma w Y'\leq A \\
& \Longleftrightarrow & (cX\sigma)(wY')\leq A \\
& \Longleftrightarrow & \mbox{there exists $(X'',Y'')\in R_A$ such that $cX\sigma\leq X''$, $wY'\leq Y''$} \\
& \Longleftrightarrow&  \mbox{there exists $(X'',Y'')\in R_A$ such that $c\leq X\sigma\rw_{incl}  X''$, $1\leq wY'\rw_{incl}  Y''$},
\end{eqnarray*}
which implies
\begin{eqnarray*}
c &\leq& (X\sigma\rw_{incl} X'')\otimes (wY'\rw_{incl}  Y'')\\
 &=& (X\sigma Y''\rw_{incl} A)\otimes (X''wY'\rw_{incl}  A) \\
&\leq& \eta_A^{\ast}((X,Y),\sigma w,(X^{\prime}, Y^{\prime})).
\end{eqnarray*}
Hence, $X\sigma w\rw_{incl}  X'\leq \eta_A^{\ast}((X,Y),\sigma w,(X^{\prime}, Y^{\prime}))$.

Therefore, the equality (\ref{eq:extension of eta}) holds.
\end{proof}

We give some properties of universal weighted automaton as follows.

\begin{proposition}\label{pro:future-past}
Let $S$ be a complete c-semiring and suppose $A\in\allseries$.
For $(X,Y)\in R_A$, we have $Past_{{\cal U}_A}(X,Y)=X$, $Fut_{{\cal U}_A}(X,Y)=Y$.
\end{proposition}

\begin{proof} For any $\theta\in\sa$,
\begin{eqnarray*}
Past_{{\cal U}_A}(X,Y)(\theta) &=& \bigvee_{(X',Y')\in R_A}J_A(X',Y')\otimes\eta^{\ast}((X',Y'),\theta,(X,Y))\\
  &=& \bigvee_{(X',Y')\in R_A}X'(\varepsilon)\otimes (X'\theta\rw_{incl} X)\\
  &=& \bigvee_{(X',Y')\in R_A}X'(\varepsilon)\otimes \bigwedge_{u\in\sa}(X'\theta(u)\rw X(u)) \\
  &\leq& \bigvee_{(X',Y')\in R_A}X'(\varepsilon)\otimes (X'\theta(\theta)\rw X(\theta))\\
  &=& \bigvee_{(X',Y')\in R_A}X'(\varepsilon)\otimes (X'(\varepsilon)\rw X(\theta)) \\
  &\leq& X(\theta).
\end{eqnarray*}
On the other hand,
\begin{eqnarray*}
Past_{{\cal U}_A}(X,Y)(\theta) &=& \bigvee_{(X',Y')\in R_A}J_A(X',Y')\otimes\eta^{\ast} ((X',Y'),\theta,(X,Y)) \\
&\geq& J_A(X_s,A)\otimes \eta^{\ast}((X_s,A),\theta,(X,Y)) \\
&=& X_s(\varepsilon)\otimes (X_s\theta\rw_{incl} X)\\
&=& X_s\theta\rw_{incl} X=\bigwedge_{u\in \sa}X_s\theta(u)\rw X(u)\\
&=& \bigwedge_{u\not=\theta}X_s\theta(u)\rw X(u)\wedge (X_s\theta(\theta)\rw X(\theta))\\
&=& \bigwedge_{u\not=\theta}(0\rw X(u))\wedge (1\rw X(\theta))\\
&=& X(\theta).
\end{eqnarray*}
Hence, $Past_{{\cal U}_A}(X,Y)=X$.

Similarly, we have $Fut_{{\cal U}_A}(X,Y)=Y$.
\end{proof}

\begin{theorem}\label{th:universal 1}
Let $S$ be a complete c-semiring and suppose $A\in\allseries$. Then we have $|{\cal U}_A|=A$.
\end{theorem}

\begin{proof}
For any $\theta\in\sa$,
\begin{eqnarray*}
|{\cal U}_A|(\theta) &=& \bigvee_{(X,Y)\in R_A}J_A(X,Y)\otimes Fut_{{\cal U}_A}(X,Y)(\theta) \\
&=& \bigvee_{(X,Y)\in R_A}X(\varepsilon)\otimes Y(\theta)\leq A(\theta).
\end{eqnarray*}

On the other hand,
\begin{eqnarray*}
|{\cal U}_A|(\theta) &=& \bigvee_{(X,Y)\in R_A}J_A(X,Y)\otimes Fut_{{\cal U}_A}(X,Y)(\theta)\\
& \geq & J_A(X_s,A)\otimes Fut_{{\cal U}_A}(X_s,A)(\theta) \\
&=& X_s(\varepsilon)\otimes A(\theta)=A(\theta).
\end{eqnarray*}
Therefore, $|{\cal U}_A|=A$.
\end{proof}

So far, we have defined two canonical weighted automata for each $A\in\allseries$, viz. the minimal DWA ${\cal M}_A$ and the universal weighted automaton ${\cal U}_A$. Both automata recognize $A$. Recall in the classical case, if $A$ is a regular language over $\Sigma$, then ${\cal M}_A$ and ${\cal U}_A$ are both finite automata. This does not hold in general for weighted automata. In the following we discuss when ${\cal U}_A$ is finite.

For $A\in\allseries$, 
write

\begin{equation}\label{eq:S-A}
S_A=\{c\rw a | c\in S, a\in Im(A)\},
\end{equation}
and let
\begin{equation}\label{eq:S-A^}
S_A^{\wedge}=\{\bigwedge X | X\subseteq S_A\}
\end{equation}

\noindent be the $\bigwedge$-sublattice of $S$ generated by $S_A$. Then it is well-known that $S_A^{\wedge}$ is finite iff $S_A$ is finite (cf. \cite{li93,davey02}).

The following proposition shows the relationship between the finiteness of ${\cal U}_A$ and the finiteness of ${\cal M}_A$.

\begin{proposition}\label{pro:finiteness}
Let $S$ be a complete c-semiring. For $A\in\allseries$, the following conditions are equivalent.
\begin{itemize}
\item
[(i)] ${\cal U}_A$ is finite, i.e., $R_A$ is finite.

\item
[(ii)] $A$ can be accepted by a finite DWA, and $S_A$ in Eq. (\ref{eq:S-A}) is finite.

\item
[(iii)] ${\cal M}_A$ is finite, i.e., $\equiv_A$ has finite index, and $S_A$ is finite.
\end{itemize}

In particular, if $S$ is finite c-semiring or  a linear order lattice, the above conditions are equivalent, and the condition ``$S_A$ is finite'' can be omitted.
\end{proposition}

\begin{proof}
By Proposition \ref{pro:dffa}, $A$ is recognized by a finite DWA iff $\equiv_A$ has finite index. It remains to show that $A$ has finite factorizations iff $S_A$ is finite.

Assume that $A$ can be recognized by a finite DWA and $S_A$ is finite. This implies that $\equiv_A$ has finite index, i.e., $Q_A$ is a finite set. If $u_1\equiv_A u_2$, then $A(u_1v)=A(u_2v)$ for any $v\in\sa$.
It follows that
\begin{eqnarray*}
X(u_1)=A/Y(u_1) &=& \bigwedge\{Y(v)\rw A(u_1v) | v\in\sa\} \\
&=& \bigwedge\{Y(v)\rw A(u_2v) | v\in\sa\} \\
&=& A/Y(u_2)=X(u_2).
\end{eqnarray*}
Thus $X$ induces a unique mapping from $Q_A=\sa/\equiv_A$ into $S_A^{\wedge}$. Since $Q_A$ is finite and $S_A$ is finite, the latter implies that the set $S_A^{\wedge}$ is also finite. Then the set of the mappings from $Q_A$ into $S_A^{\wedge}$ is also finite. Therefore, $R_A$ is finite.

On the other hand, suppose $R_A$ is finite. We first show ${\cal M}_A$ is a finite DWA.  Note that by Theorem \ref{th:universal 1}, $A$ is accepted by the finite universal weighted automaton ${\cal U}_A$. By Proposition \ref{pro:m-minimal} (the proof of which is independent to this proposition), ${\cal M}_A$ is a sub-automaton of ${\cal U}_A$. Therefore, $A$ can be accepted by a finite DWA.

To end the proof, we show $S_A$ is finite. We prove this by contradiction. Suppose $S_A$ is infinite. Since $Im(A)$ is finite, there is $a\in Im(A)$ such that the subset $S_1=\{c\rw a | c\in S\}$ of $S_A$ is infinite. Assume that $A(u)=a$ for $u\in\sa$. For any $c\in S$, define a series $Y_c\in\allseries$ as, $Y_c(w)=c$ if $w=\varepsilon$ and $Y_c(w)=0$ otherwise. Consider the right residual $A/Y_c$ , by a simple calculation, we have  $A/Y_c(u)=\bigwedge\{Y_c(v)\rw A(uv) | v\in\sa\}=Y_c(\varepsilon)\rw A(u)=c\rw a$. It follows that the set $\{A/Y_c | c\in S\}$ is infinite, then $rR(A)$ is infinite, and thus $R_A$ is infinite. A contradiction. Therefore $S_A$ is finite.
\end{proof}

We give two examples to illustrate the construction of the universal weighted automaton of a formal power series.

\begin{example}\label{ex:finite}

{\rm Assume $\Sigma=\{a,b\}$, $S=(\mathbb{N}\cup\{\infty\},\max,\min,0,\infty)$, where $S$ is a linear order lattice under the natural order of the integer numbers. Consider the formal power series  $A\in\allseries$ defined as follows,
\begin{displaymath}
{A(\theta)= \left\{ \begin{array}{ll}
$2$, & \textrm{if $\theta\in \sa ab \sa$}\\
$1$, & \textrm{otherwise.}\\
\end{array} \right.}
\end{displaymath}
$A$ can be recognized by the DWA ${\cal A}=(Q,\Sigma,\delta,q_0,F)$  presented in Figure 1, where $F=1\diagup q_0+1\diagup q_1+2\diagup q_2$.

\begin{figure}[ptb]
\begin{center}
\includegraphics[width=0.5\textwidth]{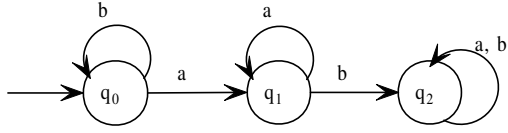}
\end{center}
\caption{The DWA recognizes the formal power series $A$}
 \label{fig:1}
\end{figure}

 The set of factorizations of
$A$ is $R_A=\{u_i|u_i=(X_i,Y_i),i=1,2,3,4.\}$, where
$X_1(\theta)=Y_2(\theta)=1$, $X_2(\theta)=Y_1(\theta)=2$, for all
$\theta \in \Sigma^*$;
\begin{displaymath}
{X_3(\theta)= \left\{ \begin{array}{ll}
$2$, & \textrm{if $\theta\in \sa a \sa$}\\
$1$, & \textrm{otherwise.}\\
\end{array} \right.},
\end{displaymath}
\begin{displaymath}
{X_4(\theta)=Y_3(\theta)= \left\{ \begin{array}{ll}
$2$, & \textrm{if $\theta\in \sa ab \sa$}\\
$1$, & \textrm{otherwise.}\\
\end{array} \right.},
\end{displaymath}
\begin{displaymath}
{Y_4(\theta)= \left\{ \begin{array}{ll}
$2$, & \textrm{if $\theta\in \sa b \sa$}\\
$1$, & \textrm{otherwise.}\\
\end{array} \right.}.
\end{displaymath}

\begin{figure}[ptb]
\begin{center}
\includegraphics[width=0.5\textwidth]{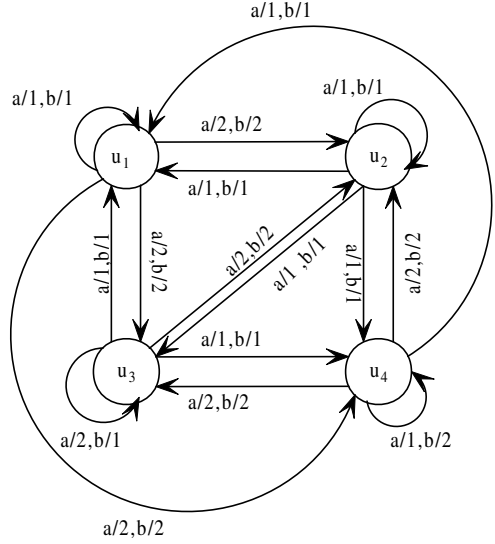}
\end{center}
\caption{The universal weighted automaton ${\cal U}_A$ of the formal power series $A$}
 \label{fig:2}
\end{figure}


By definition, the
universal weighted automaton of $A$ is ${\cal
U}_A=(R_A,\Sigma,\eta_A$, $J_A,G_A)$, where\footnote{We write $x\diagup u_i$ for the value $x$ of $u_{i}$ in the given $S$-subset.}
\begin{itemize}
\item  [-] $J_A=1\diagup u_1+2\diagup
u_2+1\diagup u_3+1\diagup u_4$,

\item [-] $G_A=2\diagup u_1+1\diagup
u_2+1\diagup u_3+1\diagup u_4$,

\item [-] $\eta_A(u_i,x,u_j)$ is either 2 or 1, as shown in Figure 2.
\end{itemize}
}
\end{example}

\begin{example}\label{ex:infinite}

{\rm Assume $\Sigma=\{a\}$, $S=(\mathbb{N}\cup\{\infty\},\min,+,\infty,0)$ is the tropical semiring. Consider the formal power series $A\in\allseries$ defined as follows
\begin{displaymath}
{A(a^k)= \left\{ \begin{array}{ll}
$0$, & \textrm{if $k=0$}\\
k-1, & \textrm{if $k>0$.}\\
\end{array} \right.}
\end{displaymath}

By Proposition \ref{pro:dffa}, $A$ can not be recognized by any finite DWA. However, as a regular series, $A$ can be recognized by a finite weighted automaton ${\cal
B}=(\{q_0,q_1\},\Sigma,\eta, \{q_0\},\{q_1\})$, where $\eta(q_0,a,q_1)=0$, $\eta(q_1,a,q_1)=1$. The universal weighted automaton ${\cal U}_A$ has infinite states
$\{(X_i,Y_i)\}_{i=0}^{\infty}$, where

\begin{center}

$Y_0=A, X_0(a^k)=k$;

\end{center}

\noindent and for $i>0$,

\begin{displaymath}
{X_i(a^k)= \left\{ \begin{array}{ll}
\infty, & \textrm{if $k=0$}\\
max(k-i,0), & \textrm{if $k>0$,}\\
\end{array} \right.}
\end{displaymath}

\begin{center}

$Y_i(a^k)=k$, for any $k\geq 0$.

\end{center}

For any $i,j$, we have $\eta_A((X_i,Y_i),a,(X_j,Y_j))=0$ if $i\leq j+1$ and $\infty$ otherwise, $J_A=(X_0,Y_0)$ and $G_A=R_A\setminus J_A$. In this case, ${\cal U}_A$ is not a finite weighted automaton.}

\end{example}

\section {The Universality of the Universal Weighted Automaton}
In this section, we show the weighted automaton ${\cal U}_A$ defined in the previous section satisfies the following universal property.
\begin{definition}
\label{dfn:universal-property}
Suppose $A\in\allseries$ and $\cal U$ is a weighted automaton that recognizes $A$. We say $\cal U$ satisfies the \emph{universal property} if there exists a morphism from $\cal B$ to $\cal U$ for any weighted automaton $\cal B$ such that $|{\cal B}|\leq A$.
\end{definition}

To demonstrate the universality of ${\cal U}_A$, we introduce a canonical mapping for each weighted automaton $\cal B$ that recognizes a subset of $A$.
\begin{definition}
Let $S$ be a complete c-semiring. Suppose $A\in\allseries$ and ${\cal B}=(P,\Sigma,\eta,J,G)$ is a weighted automaton such that $|{\cal B}|\leq A$. We define $\varphi_{\cal B}:
P\rw R_A$ by $\varphi_{\cal B}(p)=(X_p,Y_p)$ for any $p\in P$, where
\begin{equation}
\mbox{$Y_p=Past_{{\cal B}}(p)\backslash A$ and $X_p=A/Y_p$.}
\end{equation}
We call $\varphi_{\cal B}$ the canonical mapping from $\cal B$ to ${\cal U}_A$.
\end{definition}

The following lemma shows that  $\varphi$ is a morphism.
\begin{lemma}\label{lemma:universality}
Let $S$ be a complete c-semiring. Suppose $A\in\allseries$. If ${\cal B}=(P,\Sigma,\eta,J,G)$ is a weighted automaton such that $|{\cal B}|\leq A$, then the canonical mapping $\varphi_{\cal B}$ is a morphism from ${\cal B}$ into ${\cal U}_A$.
\end{lemma}
\begin{proof}
If $\varphi_{\cal B}(p)=(X_p,Y_p)$, then it is obvious that $Past_{{\cal B}}(p)\leq X_p$ and $Fut_{{\cal B}}(p)\leq Y_p$. It follows that $J_A(\varphi_{\cal B}(p))=X_p(\varepsilon)\geq Past_{{\cal B}}(p)(\varepsilon)\geq J(p)$ and $G_A(\varphi_{\cal B}(p))=Y_p(\varepsilon)\geq Fut_{{\cal B}}(p)(\varepsilon)\geq G(p)$.

For the remainder part, notice that $\eta(p,\sigma,q)Past_{{\cal B}}(p)\sigma \leq Past_{{\cal B}}(q)$. Then
$$\eta(p,\sigma,q)Past_{{\cal B}}(p)\sigma Y_q\leq Past_{{\cal B}}(q)Y_q\leq A,$$
i.e., $Past_{{\cal B}}(p)(\eta(p,\sigma,q)\sigma Y_q)\leq A$. Hence, $\eta(p,\sigma,q)\sigma Y_q\leq Y_p$. Then it follows that $$\eta(p,\sigma,q)\leq \sigma Y_q\rw_{incl} Y_p=\eta_A(\varphi(p),\sigma,\varphi(q)).$$
Therefore, $\varphi_{\cal B}$ is a morphism from ${\cal B}$ into ${\cal U}_A$.
\end{proof}
The universality of ${\cal U}_A$ follows immediately.
\begin{theorem}\label{th:universality}
Let $S$ be a complete c-semiring. For $A\in\allseries$, the universal weighted automaton ${\cal U}_A$ satisfies the universal property, i.e. there is a morphism to ${\cal U}_A$ from any weighted automaton ${\cal B}$ such that $|{\cal B}|\leq A$.
\end{theorem}


The following lemma establishes the connection between the formal power series accepted by two weighted automata connected by a morphism.


\begin{lemma}\label{lemma:morphism1}
Let $S$ be a complete c-semiring. Suppose ${\cal A}=(Q,\Sigma,\delta,I,F)$ and ${\cal B}=(P,\Sigma,\eta,J,G)$ are two weighted automata. If $\varphi$ is a morphism from ${\cal A}$ into ${\cal B}$, then we have
\begin{equation}
\mbox{$Past_{{\cal A}}(q)\leq Past_{{\cal B}}(\varphi(q))$, $Fut_{{\cal A}}(q)\leq Fut_{{\cal B}}(\varphi(q))$,}
\end{equation}
for any $q\in Q$, and thus, $|{\cal A}|\leq |{\cal B}|$. If $\varphi$ is a strong homomorphism, then $|{\cal A}|=|{\cal B}|$.
\end{lemma}

\begin{proof}
For any $\theta\in\sa$,
\begin{eqnarray*}
Past_{{\cal A}}(q)(\theta) &=& \bigvee_{i\in Q}I(i)\otimes \ds(i,\theta,q) \\
 &\leq& \bigvee_{i\in Q}J(\varphi(i))\otimes \eta^{\ast}(\varphi(i),\theta,\varphi(q)) \\
 &\leq& \bigvee_{j\in P}J(j)\otimes \eta^{\ast}(j,\theta,\varphi(q)) \\
 &=& Past_{{\cal B}}(\varphi(q))(\theta). \\
Fut_{{\cal A}}(q)(\theta) &=& \bigvee_{t\in Q}\ds(q,\theta,t)\otimes F(t) \\
&\leq&  \bigvee_{t\in Q}\eta^{\ast}(\varphi(q),\theta,\varphi(t))\otimes G(\varphi(t)) \\
&\leq& \bigvee_{p\in P}\eta^{\ast}(\varphi(q),\theta,p)\otimes G(p) \\
&=& Fut_{{\cal B}}(\varphi(q))(\theta)).
\end{eqnarray*}

Hence, $Past_{{\cal A}}(q)\leq Past_{{\cal B}}(\varphi(q))$ and $Fut_{{\cal A}}(q)\leq Fut_{{\cal B}}(\varphi(q))$. Then it follows that $|{\cal A}|\leq |{\cal B}|$.

If $\varphi$ is a strong homomorphism, then it can be easily verified that $\eta^{\ast}(\varphi(q),\theta,p)=\bigvee\{\ds(q,\theta,r)| \varphi(r)=p\}$ for any $\theta\in\sa$.
Then it follows that
\begin{eqnarray*}
|{\cal B}|(\theta) &=& \bigvee_{j,p\in P}J(j)\otimes \eta^{\ast}(j,\theta,p)\otimes G(p) \\
 &=& \bigvee_{i\in Q,p\in P}I(i)\otimes  \eta^{\ast}(\varphi(i),\theta,p)\otimes G(p) \\
 &=& \bigvee_{i\in Q, \varphi(r)=p}I(i)\otimes\ds(i,\theta,r)\otimes G(\varphi(r)) \\
 &=& \bigvee_{i,r\in Q}I(i)\otimes \ds(i,\theta,r)\otimes F(r)=|{\cal A}|(\theta).
\end{eqnarray*}
Hence, $|{\cal A}|=|{\cal B}|$.
\end{proof}

Suppose ${\cal A}=(Q,\Sigma,\delta,I,F)$ is a weighted automaton. We say two states $p$ and $q$ in $Q$ are \emph{mergible} in ${\cal A}$ if there exist a weighted automaton ${\cal B}=(P,\Sigma,\eta,J,G)$ that accepts the same language as ${\cal A}$ and  a surjective morphism $\varphi: {\cal A}\rw {\cal B}$ such that $\varphi(p)=\varphi(q)$.

\begin{proposition}\label{pro:mergible}
Let $S$ be a complete c-semiring. Suppose $A \in \allseries$. Then there is no mergible states in the universal weighted automaton ${\cal U}_A$.
\end{proposition}

\begin{proof}
Otherwise, there is a weighted automaton ${\cal C}$ and a surjective morphism $\varphi: {\cal U}_A\rw {\cal C}$ such that $|{\cal U}_A|=| {\cal C}|=A$,
and $\varphi(X,Y)=\varphi(X',Y')=s$ for two distinct states $(X,Y)$ and $(X',Y')$ in ${\cal U}_A$. Then we have,
\begin{eqnarray*}
X=Past_{{\cal U}_A}(X,Y)\leq Past_{{\cal C}}(\varphi(X,Y))=Past_{{\cal C}}(s), \\
X'=Past_{{\cal U}_A}(X',Y')\leq Past_{{\cal C}}(\varphi(X',Y'))=Past_{{\cal C}}(s),
\end{eqnarray*}
and thus, $X\vee X'\leq Past_{{\cal C}}(s)$. Similarly, we have $Y\vee Y'\leq Fut_{{\cal C}}(s)$. Therefore,
$$(X\vee X')(Y\vee Y')\leq Past_{{\cal C}}(s)Fut_{{\cal C}}(s)\leq A.$$
This contradicts with the maximality of the factorization $(X,Y)$ and $(X',Y')$.
\end{proof}


In fact, ${\cal U}_A$ is the largest non-mergible weighted automaton.
\begin{corollary}\label{co:largest}
Let $S$ be a complete c-semiring. Suppose $A\in\allseries$. Then ${\cal U}_A$ is the largest weighted automaton among those that accept $A$ but have no mergible states.
\end{corollary}

\begin{proof}
A weighted automaton ${\cal B}$ accepting $A$ that has strictly more states than ${\cal U}_A$ is sent into ${\cal U}_A$ by a morphism which is necessarily non-injective.
\end{proof}
Moreover, ${\cal U}_A$ is the smallest `universal' weighted automaton.
\begin{proposition}\label{pro:least}
Let $S$ be a complete c-semiring. Suppose $A\in\allseries$. Then  ${\cal U}_A$ is the smallest weighted automaton among those that accept $A$ and have the universal property.
\end{proposition}

\begin{proof}
Suppose that ${\cal C}$ has the universal property with respect to the $A$. As ${\cal U}_A$ accepts $A$, there should be a morphism from ${\cal U}_A$ into ${\cal C}$. As ${\cal U}_A$ has no mergible states, this morphism should be injective: ${\cal C}$ has at least as many states as ${\cal U}_A$.
\end{proof}

Applying Theorem \ref{th:universality} to weighted automata that accepts $A$, we obtain the following corollary.

\begin{corollary}\label{co:m-minimal}
Let $S$ be a complete c-semiring. Suppose $A\in\allseries$ and ${\cal B}$ is a weighted automaton that accepts $A$. If $\cal B$ has no mergible states, then ${\cal B}$ is a sub-automaton of ${\cal U}_A$.
\end{corollary}

\begin{proof}
By Theorem \ref{th:universality}, the canonical mapping $\varphi$ from ${\cal B}$ to ${\cal U}_A$ is a morphism. Because ${\cal B}$ has no mergible states, we know $\varphi$ must be one-to-one. Therefore, ${\cal B}$ is a sub-automaton of ${\cal U}_A$.
\end{proof}

It is not difficult to show that any minimal weighted (determinate or non-determinate) automaton that accepts $A$ has no mergible states. The above corollary then suggests a simple way for searching the minimal (non-determinate) weighted automaton that accepts $A$: It suffices to check the sub-automaton of the universal automaton ${\cal U}_A$ which accepts $A$ and has minimal states. The following proposition shows that ${\cal M}_A$, the minimal DWA that accepts $A$, is also a sub-automaton of ${\cal U}_A$.




\begin{proposition}\label{pro:m-minimal}
Let $S$ be a complete c-semiring. Suppose $A\in \allseries$. Then ${\cal M}_A$ is a sub-automaton of ${\cal U}_A$.
\end{proposition}
\begin{proof}
By Corollary~\ref{co:m-minimal}, we need only show that ${\cal M}_A$ has no mergible states.

Recall the minimal DWA that accepts $A$ is ${\cal M}_A=(Q,\Sigma,\delta,q_0,F)$, where
\begin{itemize}
\item [-] $Q=\{u^{-1}A | u\in\sa\}$,
\item [-] $\delta(u^{-1}A,\sigma)=(u\sigma)^{-1}A$,
\item [-] $\ds(u^{-1}A,v)=(uv)^{-1}A$,
\item [-] $q_0=(\varepsilon)^{-1}A=A$,
\item [-] $F:Q\rw S$ is $F(u^{-1}A)=A(u)$.
\end{itemize}
Then for any state $u^{-1}A\in Q$, $\ds(A,u)=u^{-1}A$. We show that there are no mergible states in ${\cal A}$. Otherwise, there are two distinct states $u^{-1}A$, $v^{-1}A$ in $Q$, but there exists another weighted automaton ${\cal B}$ and a morphism $\varphi$ from ${\cal A}$ into ${\cal B}$ such that ${\cal B}$ is the morphic image of ${\cal A}$, $|{\cal A}|=|{\cal B}|$ and $\varphi(u^{-1}A)=\varphi(v^{-1}A)$.

Since $u^{-1}A\not=v^{-1}A$, there exists $w\in\sa$ such that $u^{-1}A(w)\not=v^{-1}A(w)$, i.e., $A(uw)\not=A(vw)$. Note that
\begin{eqnarray*}
A(uw)=|{\cal A}|(uw)=F(\ds(A,uw))=F(\ds(\ds(A,u),w))=F(\ds(u^{-1}A,w)),\\
A(vw)=|{\cal A}|(vw)=F(\ds(A,vw))=F(\ds(\ds(A,v),w))=F(\ds(v^{-1}A,w)).
\end{eqnarray*}
Then
\begin{eqnarray*}
|{\cal B}|(uw) &=& \varphi(F)(\delta^{\ast}(\varphi(A),uw)) \\
 &=& \varphi(F)(\varphi(\ds(A,uw)))\\
&=&\varphi(F)(\varphi(\ds(u^{-1}A,w))) \\
&=& \varphi(F)(\delta^{\ast} (\varphi(u^{-1}A),w))) \\
&=& \varphi(F)(\delta^{\ast}(\varphi(v^{-1}A),w))) \\
&=& \varphi(F)(\delta^{\ast}(\varphi(A),vw))\\
&=& |{\cal B}|(vw).
\end{eqnarray*}

Since $|{\cal B}|=|{\cal A}|=A$, it follows that $A(uw)=A(vw)$, a contradiction occurs.
Therefore, ${\cal A}$ has no mergible states.
\end{proof}

\section{Construction of the Universal Weighted Automaton}

In general, it is not effective to construct all factorizations of $A$. Suppose  ${\cal A}=(Q,\Sigma,\delta,q_0,F)$ is an arbitrary DWA accepting $A$. In this section, we give an effective method to construct ${\cal U}_A$ by using the DWA $\cal A$. Let $l_A$ be the $\bigvee$-sublattice generated by $S_A^{\wedge}$ as defined in Eq.(\ref{eq:S-A}), i.e.,
\begin{equation}\label{eq:l-A}
l_A=\{\bigvee X | X\subseteq S_A^{\wedge}\}.
\end{equation}

\noindent It is well known that $l_A$ is finite iff $S_A^{\wedge}$ is finite (cf.\cite{li93,davey02}) iff $S_A$ is finite. Write $Q_1=l_A^Q$. If ${\cal A}$ is finite and $S_A$ is finite, then $Q_1$ is also finite. We construct
a weighted automaton ${\cal A}_1=(Q_1,\Sigma,\eta,J,G)$ as follows:
\begin{eqnarray*}
J(f) &=&f(q_0),\\
G(f) &=& f\rw_{incl} F,\\
\eta(f,\sigma,g) &=& f\sigma\rw_{incl} g,
\end{eqnarray*}
where $f\sigma:Q\rw l_A$ is defined by $f\sigma (q)=\bigvee\{f(p) | \delta(p,\sigma)=q\}$.

We next define a mapping $\varphi$ from ${\cal A}_1$ to ${\cal U}_A$. To this end, we first establish the correspondence between weighted states and factorizations of $A$.

\begin{proposition}
\label{prop: Xf&Yf}
Let $S$ be a complete c-semiring. Suppose $A\in\allseries$ and ${\cal A}=(Q,\Sigma,\delta,q_0,F)$ is an arbitrary DWA accepting $A$. Then  $(X_f,Y_f)$ is a factorization of $A$ for any weighted state $f: Q\rw l_A$, where
\begin{eqnarray}
\label{eq:Y_f}
Y_f &=& \bigwedge_{q\in Q} f(q)\rw Fut_{{\cal A}}(q),\\
\label{eq:X_f}
X_f &=& A/Y_f,
\end{eqnarray}
and, for any $\theta\in\sa$,
\begin{eqnarray}
(f(q)\rw Fut_{{\cal A}}(q))(\theta) &=& f(q)\rw Fut_{{\cal A}}(q)(\theta).
\end{eqnarray}
Therefore, the mapping $\varphi$ defined by
$\varphi(f)=(X_f,Y_f)$
 is a mapping from $Q_1=l_A^Q$ to $R_A$.
 \end{proposition}
 \begin{proof}
Without loss of generality, we assume that ${\cal A}$ is accessible. Since ${\cal A}$ is a DWA, for any $q\in Q$, there exists $u\in \sa$ such that $\ds(q_0,u)=q$. Moreover, if $\ds(q_0,v)=q$ for another $v\in \sa$, then $A(uw)=A(vw)$ for any $w\in \sa$. By this observation, we have, for any $v\in\sa$,
\begin{eqnarray*}
Y(v) &=& \bigwedge_{q\in Q}f(q)\rw Fut_{{\cal A}}(q)(v)\\
& = & \bigwedge_{u\in\sa,\ds(q_0,u)=q}f(q)\rw A(uv) \\
& =& \bigwedge_{u\in\sa}f(\ds(q_0,u))\rw A(uv).
\end{eqnarray*}

If we let $X'(u)=f(\ds(q_0,u))$ for any $u\in \sa$, then we obtain a series $X':\sa\rw l_A$ such that $Y=X'\backslash A$.
\end{proof}

The mapping $\varphi: l_A^Q \rw R_A$ is also onto.
\begin{proposition}\label{pro:facor-state1}
Let $S$ be a complete c-semiring. Suppose $A\in\allseries$ and ${\cal A}=(Q,\Sigma,\delta,q_0,F)$ is a DWA that accepts $A$. For any $(X,Y)\in R_A$, there is a weighted state $f:Q\rw l_A$ such that $Y=Y_f$.
\end{proposition}

\begin{proof}
Define a weighted state $f:Q\rw l_A$ as, for any $q\in Q$,
$$f(q)=\bigvee\{X(u) | \ds(q_0,u)=q\}.$$

\noindent By the proof of Proposition \ref{pro:finiteness}, $X$ induces a unique mapping from $Q_A=\sa/\equiv_A$ into $S_A^{\wedge}$, so $f$ is well-defined.

We show $Y=Y_f$ in the following.

For any $\theta\in\sa$, we have
\begin{eqnarray*}
Y_f(\theta) &=& \bigwedge_{q\in Q}f(q)\rw Fut_{{\cal A}}(q)(\theta)\\
&=& \bigwedge_{q\in Q}(\bigvee\{X(u) | \ds(q_0,u)=q\})\rw F(\ds(q,\theta)\\
&=& \bigwedge_{q\in Q, \ds(q_0,u)=q}X(u)\rw F(\ds(q,\theta)) \\
&=& \bigwedge_{u\in\sa}X(u)\rw A(u\theta) \\
&=& X\backslash A(\theta) \\
&=& Y(\theta).
\end{eqnarray*}
Hence, $Y=Y_f$.
\end{proof}

 Furthermore, we have

\begin{proposition}\label{pro:strong homo}
Let $S$ be a complete c-semiring. Suppose $A\in \allseries$. The mapping $\varphi$ defined in Proposition~\ref{prop: Xf&Yf} is a strong homomorphism from weighted automaton ${\cal A}_1=(Q_1,\Sigma,\eta,J,G)$ \emph{onto} the universal weighted automaton ${\cal U}_A$, and thus $|{\cal A}_1|=|{\cal U}_A|=A$.

\end{proposition}
\begin{proof}
See Appendix B.
\end{proof}

Define an equivalence relation $\sim$ on $Q_1$ as follows:
\begin{equation}
\label{eq:sim-relation}
\mbox{$f\sim g$ iff $\varphi(f)=\varphi(g)$}
\end{equation}
It is clear that $f\sim g$ iff $Y_f=Y_g$. Using this equivalence relation, we obtain a quotient weighted automaton from ${\cal A}_1$, denoted by
${\cal A}^{\prime}$, which is isomorphic to ${\cal U}_A$.

\begin{corollary}\label{co:construction}
Let $A\in\allseries$ and ${\cal A}_1$ be as in Proposition~\ref{pro:strong homo}. Suppose ${\cal A}^{\prime}$ is the quotient weighted automaton of ${\cal A}_1$ modulo the equivalent relation $\sim$ on $Q_1$. Then ${\cal A}^\prime$ is isomorphic to ${\cal U}_A$.
\end{corollary}

 Once the DWA ${\cal A}$ is finite and $S_A$ is finite (this condition can be guaranteed if $S$ is finite or $S$ is a linear-order lattice as declared in Proposition \ref{pro:finiteness}), the equivalence $\sim$ defined by Eq.(\ref{eq:sim-relation}) can be effectively constructed. This is because, $\ds(q,\theta)$ takes at most $|Q|=n$ states, i.e., the set $\{\ds(q,\theta) | \theta\in\sa\}$ as a subset of $Q$ has at most $n$ states, the corresponding $Fut_{{\cal A}}(q)=G_A(\ds(q,\theta))$ has at most $n$ values. Therefore, it is sufficient to check these states in Eq.(\ref{eq:sim-relation}). In turn, it is sufficient to check those $\theta\in\sa$ with  $|\theta|<n$ in Eq.(\ref{eq:sim-relation}). Hence, the equivalence relation $\sim$ is decidable
and the weighted automaton ${\cal A}^{\prime}$ can be effectively constructed.

We next give one example.

\begin{example}\label{ex:finite-state}
{\rm Consider the formal power series $A$ in Example \ref{ex:finite}, which is recognized by the finite DWA as shown in Figure 1. Then the weighted automaton ${\cal A}'=(Q_1,\Sigma,\eta,J,G)$, where
$Q_1=\{f_1,f_2,f_3,f_4\}$ with $f_1(q)=1$ and $f_2(q)=2$ for any $q\in Q$,
\begin{displaymath}
{f_3(q)= \left\{ \begin{array}{ll}
$2$, & \textrm{$q = q_1,q_2$}\\
$1$, & \textrm{$q=q_0.$}\\
\end{array} \right.}
\end{displaymath}
\begin{displaymath}
{f_4(q)= \left\{ \begin{array}{ll}
$2$, & \textrm{$q = q_2$}\\
$1$, & \textrm{$q=q_0,q_1.$}\\
\end{array} \right.}
\end{displaymath}
and
\begin{eqnarray*}
J &=& 1\diagup f_1+2\diagup f_2+1\diagup f_3+1\diagup f_4,\\
G &=& 2\diagup f_1+1\diagup f_2+1\diagup f_3+1\diagup f_4;
\end{eqnarray*}
and
$\eta(f_i,x,f_j)$ is either 2 or 1 (see Figure 3). Clearly, ${\cal A}'$ is isomorphic to ${\cal U}_A$.

 \begin{figure}[ptb]
\begin{center}
\includegraphics[width=0.5\textwidth]{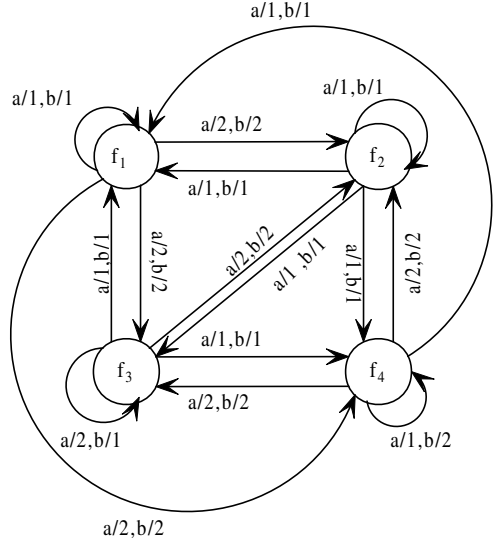}
\end{center}
\caption{The weighted automaton  ${\cal A}'$}
 \label{fig:2}
 \end{figure} }
\end{example}

We next give a detailed examination of the equivalent relation $\sim$.

The correspondence $\varphi$ between the weighted states and the factorizations of $A$ stated in Proposition~\ref{prop: Xf&Yf} may be not one-to-one. There may have more than one weighted states correspond to a given factorization. However, the following proposition asserts that there exists a largest weighted states.

\begin{proposition}\label{pro:largest state}
Let $S$ be a complete c-semiring. Suppose $A\in \allseries$ and ${\cal A}=(Q,\Sigma,\delta,q_0,F)$ is a DWA that accepts $A$. If $Y=Y_{g_i}$ for all weighted states $g_i: Q\rw l_A$ ($i\in I$), then $Y=Y_{\bigvee_{i\in I}g_i}$. Therefore, for each right residual $Y$, there exists a largest $h:Q\rw l_A$ such that $Y=Y_h$.
\end{proposition}

By the above propositions, we know for each right residual $Y$ of $A$ there exists a (unique) largest weighted state $h:Q\rw l_A$ such that $Y=Y_h$. Furthermore, suppose $(X,Y)\in R_A$. Then $h$ is defined by $h(q)=\bigvee\{X(u) | \ds(q_0,u)=q\}$ as implied in the proofs of Propositions \ref{prop: Xf&Yf} and \ref{pro:facor-state1}. In general, for a weighted state $f:Q\rw l_A$, the structure of the largest weighted state $h$ such that $Y_f=Y_h$
is unclear. But we have the following estimation.


\begin{proposition}\label{pro:largest element1}
Let $S$ be a complete c-semiring. Suppose ${\cal A}=(Q,\Sigma,\delta$, $q_0,F)$ is a weighted automaton, and $f:Q\rw l_A$ a weighted state. If $\{u_i\}_{i\in I}\subseteq\sa$ and  $Y_f=Y_{\bigwedge_{i\in I}u_i\circ F}$,
then $f\leq \bigwedge_{i\in I}u_i\circ F$, where $u\circ F(q)=F(\ds(q,u))$.
\end{proposition}

\begin{proof}
Let $g=\bigwedge_{i\in I}u_i\circ F$. Then for each $j\in I$ we have
 $$Y_g(u_j)=\bigwedge_{q\in Q}g(q)\rw Fut_{{\cal A}}(q)(u_j)=\bigwedge_{q\in Q}(\bigwedge_{i\in I}u_i\circ F)(q)\rw (u_j\circ F)(q)=1.$$
Thus $\bigwedge_{q\in Q}f(q)\rw Fut_{{\cal A}}(q)(u_j)=1$. It follows that $f(q)\rw Fut_{{\cal A}}(q)(u_j)=1$, i.e., $f(q)\leq Fut_{{\cal A}}(q)(u_j)=u_j\circ F(q)$ for any $q\in Q$.
Therefore, $f\leq u_j\circ F$ for any $j$, hence $f\leq \bigwedge_{i\in I}u_i\circ F$.
\end{proof}

Recall that if $S$ is the two elements Boolean algebra $\{0,1\}$, then $\{\bigwedge_{i\in I}u_i\circ F | u_i\in\sa\}$ forms the whole set of those largest elements  (cf.\cite{lombardy08}). It is still unclear whether each largest weighted state can be represented as the intersection of $u_i\circ F$ for a set of $u_i$.


\section{A Comparison of Quotients and Residuals}
For a formal power series $A$, we have introduced the notions of (left and right) quotients and (left and right) residuals of $A$. As can be seen in the above discussion, these two kinds of operations have different behaviors when considering their algebraic and language properties. Especially, with the associated weighted automata, ${\cal B}_A$ and ${\cal U}_A$ have different structure, although they are equivalent as language recognizers. For example, there exists series $A$ such that ${\cal B}_A$ is infinite but ${\cal U}_A$ is finite.

In this section, we consider the order relation between the quotients and the residuals of a same series $A$. Because the duality between left and right quotients (residuals), we need only compare the left quotient $X^{-1}A$ and  the left residual $X\backslash A$ of $A$ by a series $X$. Our results show that all the four possibilities are possible.

First, if $X$ is a word $u\in\sa$ or $\sum_{u\in\sa}X(u)=1$ and $A=0$, then it is obvious that $X^{-1}A=X\backslash A$.

Second, let $S=\{0,1\}$ be the two element Boolean algebra. Then for any languages $X,A\subseteq\sa$ we  have $$X\backslash A=\bigcap_{u\in X} u^{-1}A\subseteq \bigcup_{u\in X}u^{-1}A=X^{-1}A.$$
Moreover, $X^{-1}A=X\backslash A$ iff $u\equiv_A v$ for any $u,v\in X$; but if there exists $u,v\in X$ such that $u\not\equiv_A v$, then the above inclusion is strict.

Third, let $S$ be the tropical semiring $(\mathbb{N}+\cup\{\infty\}$, $\min,+,\infty,0)$. For $A\in\allseries$, $u\in\sa$, and $k\in S$, we have $(ku)^{-1}A=k(u^{-1}A)\leq_S u^{-1}A$, while $(ku)\backslash A >_S u^{-1}A$ if $k\not=0$ and $k\not=\infty$. Therefore $(ku)\backslash A>_S(ku)^{-1}A$ if $k\not=0$ and $k\not=\infty$.

Fourth, let $S$ be the tropical semiring $(\mathbb{N}+\cup\{\infty\},\min,+,\infty,0)$. Take $\Sigma=\{a,b\}$. Consider the formal power series $A\in\allseries$ defined as follows
\begin{displaymath}
{A(\theta)= \left\{ \begin{array}{ll}
0, & \textrm{if $\theta=ba$ or $\theta=bb$}\\
10, & \textrm{if $\theta=aa$}\\
3, & \textrm{if $\theta=ab$}\\
\infty, & \textrm{otherwise.}\\
\end{array} \right.}
\end{displaymath}
Let $X=min((4+a),2+b)$. We show $X^{-1}A$ and $X\backslash A$ are incomparable. By
\begin{eqnarray*}
X^{-1}A(a) &=& min(4+A(aa),2+A(ba))=2, \\
X^{-1}A(b) &=& min(4+A(ab), 2+A(bb))=2, \\
X\backslash A(a) &=& max(4\rw A(aa), 2\rw A(ba))=6, \\
X\backslash A(b) &=& max(4\rw A(ab), 2\rw A(bb))=0,
\end{eqnarray*}
we know $X^{-1}A(a)>_S X\backslash A(a)$, but $X^{-1}A(b)<_S X\backslash A(b)$. Therefore $X^{-1}A$ and $X\backslash A$ are  incomparable.

It is still unclear when (i.e. for what kind of $S$ and $A\in\allseries$) we will have an uniform order relation between $X^{-1}A$ and $X\backslash A$.

\section  {Conclusions}
In this paper, we have defined the quotient and residual operations for formal power series. The algebraic and closure properties under these operations were discussed. Our results show that most nice properties are kept in formal power series. Moreover, we introduced two canonical weighted automata ${\cal M}_A$ and ${\cal U}_A$ for each formal power series $A$ using the quotients and, respectively, residuals of $A$. It was shown that ${\cal M}_A$ is the minimal DWA of $A$, and ${\cal U}_A$ is the universal weighted automaton of $A$ which contains as a sub-automaton any weighted automaton that accepts $A$ but has no mergible states. In particular, any minimal weighted (deterministic or non-deterministic) automaton of $A$ is a sub-automaton of ${\cal U}_A$. This suggests an efficient way to find (approximations of) the minimal weighted NFA of $A$. Last but not least, we also showed, under a rather weak restriction, that ${\cal U}_A$ is finite iff ${\cal M}_A$ is finite, and that ${\cal U}_A$  can be effectively constructed when we have a finite DWA that recognizes $A$.

There are still several open problems left unsolved. Suppose $X$ and $A$ are two formal power series. Is $X\backslash A$ regular (context-free) whenever $A$ is regular (context-free)? What is the precise relation between $X^{-1}A$ and $X\backslash A$? When characterizing residuals in terms of quotients by word, we require the underlying semiring to be a complete c-semiring. It is easy to see that this requirement is not necessary. Another problem then is, to what extent can we develop the related theory in a weaker semiring structure, e.g. in quantale?

Another interesting question is to develop an abstraction scheme for formal power series and weighted automata based on semiring homomorphism. Just like in soft constraint satisfaction \cite{BCR02,LiY08},  we expect that semiring homomorphisms can play an important role in approximately computing the minimal NFA and the universal weighted automaton of a formal power series.

\appendix
\section{Proof of Propositon~\ref{pro:quotient-closeness}}
\begin{proof}

(i) Suppose ${\cal A}=(Q,\Sigma,\delta,I,F)$ is a weighted automaton accepting $A$. Then, for any $\theta\in\sa$, we have
$$|{\cal A}|(\theta)=\sum_{q_0,q\in Q}I(q_0)\ds(q_0,\theta,q) F(q)=A(\theta).$$

Define another two weighted automata ${\cal A}_X=(Q,\Sigma,\delta,I_X,F)$ and ${\cal A}_Y=(Q,\Sigma$, $\delta,I,F_Y)$, where

$$I_X(q)=\sum_{u\in\sa,q_0\in Q}X(u) I(q_0)\ds(q_0,u,q),$$

$$F_Y(q)=\sum_{v\in \sa,p\in Q}\ds(q,v,p) F(p) Y(v).$$

Then

\begin{eqnarray*}
|{\cal A}_X|(\theta) &=& \sum_{q_1,q\in Q}I_X(q_1)\ds(q_1,\theta,q) F(q)\\
                              &=& \sum_{q_0,q_1,q\in Q,u\in \sa}X(u)I(q_0)\ds(q_0,u,q_1) \ds(q_1,\theta,q) F(q))\\
                              &=& \sum_{u\in\sa,q_0,q\in Q}I(q_0)\ds(q_0,u\theta,q)  F(q)\\
                              &=& \sum_{u\in\sa}X(u)|{\cal A}|(u\theta)=X^{-1}A(\theta).
\end{eqnarray*}

\begin{eqnarray*}
|{\cal A}_Y|(\theta) &=& \sum_{q_0,q\in Q}I(q_0)\ds(q_0,\theta,q) F_Y(q)\\
                              &=& \sum_{q_0,q,p\in Q,v\in \sa}I(q_0)\ds(q_0,\theta,q) \ds(q,v,p) F(p) Y(v)\\
                              &=& \sum_{q_0,p\in Q,v\in \sa}I(q_0)\ds(q_0,\theta v,p)  F(p) Y(v)\\
                              &=& \sum_{v\in\sa}|{\cal A}|(\theta v)Y(v)=AY^{-1}(\theta).
\end{eqnarray*}
Hence, $X^{-1}A=|{\cal A}_X|$, and $AY^{-1}=|{\cal A}_Y|$  are regular.

(ii) Replace weighted automaton ${\cal A}$ in the proof of (i) by a DWA, we can prove that $AY^{-1}=|{\cal A}_Y|$ is DWA-regular.

Since $A$ is DWA-regular, by Proposition~\ref{pro:dffa}, there are $r_1,\cdots,r_k\in S\{0\}$ and regular languages $L_1,\cdots,L_k$ such that $A=\sum_{i=1}^kr_iL_i$. By Proposition~\ref{pro:quotient-fps}, $X^{-1}A=\sum_{i=1}^k r_iX^{-1}L_i$. By Lemma~\ref{lemma:right-quotient}, $X^{-1}L_i=(L_i^R(X^R)^{-1})^R$, it follows that $X^{-1}L_i$ is DWA-regular for any $i$, and thus $X^{-1}A$ is DWA-regular by Proposition~\ref{pro:dffa} again.

(iii)  Let $Y_1=\sum_{w\in \Sigma^+} (Y,w)w$. Then $Y=Y(\varepsilon)\varepsilon+Y_1$ and $Y_1$ is the proper part of $Y$. It is obvious that $Y$ is regular iff $Y_1$ is regular. By Proposition \ref{pro:quotient-fps}, $AY^{-1}=A(Y(\varepsilon)\varepsilon)^{-1}+A(Y_1)^{-1}=AY(\varepsilon)+A(Y_1)^{-1}$. $AY(\varepsilon)$ is context-free since context-free languages are closed under scalar operation. To show $AY^{-1}$ is context-free, it suffices to show that $AY_1^{-1}$ is context-free.
Without loss of generality, we assume that $Y$ is proper, i.e., $Y(\varepsilon)=0$. There exist a context free-grammar $G=(V,\Sigma,P,S)$
and a weighted automaton ${\cal A}=(Q,\Sigma,\delta,q_0,\{q_f\})$ satisfying the conditions of Lemma \ref{le:proper}, such that
$|G|=A$ and $|{\cal A}|=Y$.

Construct a new weighted context-free grammar $G'=(V',\Sigma,P',S')$, where, $V'=\Sigma\cup (Q\times \Sigma\times Q)$, $S'=(q_0,S,q_f)$,
$P'$ is constructed as follows,

(ii-1) $(q_0,x,q_0)\overset{1}{\rightarrow}x$ for each $x\in \Sigma$;

(ii-2) $(q,x,q')\overset{r}{\rightarrow}\varepsilon$ if $x\in\Sigma$ and $\delta(q,x,q')=r$;

(ii-3) $(q,x,q')\overset{r}{\rightarrow}(q,y_1,q_1)(q_1,y_2,q_2)\cdots(q_{n-1},y_n,q')$ if $x\overset{r}{\rightarrow}y_1y_2\cdots y_n$ and $q_1,\cdots,q_{n-1}\in Q$.

Let $w\in\sa$. We note that 
\begin{eqnarray*}
L(G')(w) &=&
\sum\{r_1\otimes r_2\otimes \cdots\otimes r_k | (\exists \alpha_1,\cdots,\alpha_{k-1})\\
&&  \hspace*{25mm} [S'\overset{r_1}{\Rightarrow}
\alpha_1\overset{r_2}{\Rightarrow}\alpha_2\overset{r_k}{\Rightarrow}\cdots\overset{r_{k-1}}
{\Rightarrow}\alpha_{k-1}\overset{r_k}{\Rightarrow}w]\}\\
AY^{-1}(w) &=& \sum\{A(wu)Y(u) | u\in\sa\} \\
&=& \sum\{(r_{11}\otimes r_{12}\otimes \cdots\otimes  r_{1l})\otimes (\delta(q_0,\sigma_1,q_1)\otimes \cdots \otimes \delta(q_{m-1},\sigma_m,q_f))| \\
&& \hspace*{10mm}  (\exists \alpha_1,\cdots,\alpha_{l-1}) (\exists \sigma_1,\cdots,\sigma_m \in \Sigma)(\exists q_0, q_1,\cdots,q_{m-1}\in Q) \\
&&  \hspace*{30mm} \mbox{such that\ } S\overset{r_{11}}{\Rightarrow} \alpha_1\overset{r_{12}}{\Rightarrow}\cdots\overset{r_{1l}}{\Rightarrow}w\sigma_1\cdots\sigma_m\}.
\end{eqnarray*}

 To show $L(G')(w)=AY^{-1}(w)$, it suffices to show that they have the same non-zero sum-terms in $S$. That is, each non-zero sum-term $r_1\otimes r_2\otimes \cdots\otimes r_k$  in the equality of $L(G')(w)$ also appears in the equality of $AY^{-1}(w)$ as a sum-term, and vise verse.

On one hand, suppose there exists a sequence $\alpha_1,\cdots,\alpha_{k-1}$ such that
$$S'\overset{r_1}{\Rightarrow} \alpha_1\overset{r_2}{\Rightarrow}\alpha_2\overset{r_k}{\Rightarrow}\cdots\overset{r_{k-1}}
{\Rightarrow}\alpha_{k-1}\overset{r_k}{\Rightarrow}w,$$
i.e., $r=r_1\otimes r_2\otimes \cdots\otimes r_k$ is a non-zero term in $L(G')(w)$. Notice that the semiring $S$ is assumed to be commutative and a production of type (ii-3) commutes with one of the type (ii-1) or (ii-2). The above sequence of productions can be rearranged so that all the productions of type (ii-3) precede those of type (ii-1) and (ii-2). Hence we may assume
that by type (ii-3) productions
\begin{eqnarray}\label{eq:23-1}
S'  &=& (q_0,S,q_f)\overset{r_1}{\Rightarrow}\alpha_1\overset{r_2} {\Rightarrow}\cdots\overset{r_m}{\Rightarrow}\alpha_m\\
\label{eq:23-2}                                           &=& (q_0,y_1,q_1)(q_1,y_2,q_2)\cdots(q_{n-1},y_n,q'),
\end{eqnarray}
and by type (ii-1) and (ii-2) productions
\begin{equation}\label{eq:24}
(q_0,y_1,q_1)(q_1,y_2,q_2)\cdots(q_{n-1},y_n,q')\overset{r_{m+1}}{\Rightarrow}\alpha_{m+1}\overset{r_{m+2}}{\Rightarrow}
\cdots\overset{r_k}{\Rightarrow}w.
\end{equation}

Since the induction (\ref{eq:24}) is by type (ii-1) and (ii-2), each $y_i$ is in $\Sigma$. Since every term in the induction (\ref{eq:23-2}) corresponds
to a production of $P$, it follows that there exists $\beta_1,\cdots,\beta_m$ such that
$S\overset{r_1}{\Rightarrow}\beta_1\overset{r_2}
{\Rightarrow}\cdots\overset{r_m}{\Rightarrow}\beta_m
=y_1y_2\cdots y_n$ is an induction in $G$. Furthermore, for each $i\leq n$, we have either $q_{i-1}=q_i=q_0$
or $\delta(q_{i-1},y_i,q_i)=r_t$ for some $m+1\leq t\leq k$. Let $j$ be the largest integer such that $q_j=q_0$
and let $w=y_1\cdots y_j$, $u=y_{j+1}\cdots y_n$.  Because $\delta(q,\sigma,q_0)=0$ for any $q\in Q$, it follows that $q_0=q_1=\cdots=q_j$ and $q_{i+1}\not= q_0$ for any $i\geq j$. Omitting the weight $1$ by using the type (ii-1) productions, i.e., $r_{m+1}=\cdots=r_{m+j}=1$, it follows
that the leaving terms, after rearranging, is
$$r_{m+j+1}=\delta(q_j,y_{j+1},q_{j+1}), \cdots, r_k=\delta(q_{n-1},y_n,q'=q_f)$$ and
$$r_{m+j+1}\otimes \cdots \otimes r_k=r_{m+1}\otimes \cdots\otimes  r_k.$$
This shows that the term $r_1\otimes r_2\otimes \cdots\otimes  r_k=(r_1\otimes\cdots\otimes r_m)\otimes (r_{m+1}\otimes\cdots\otimes r_k)=(r_1\otimes\cdots r_m)\otimes (r_{m+j+1}\otimes\cdots\otimes r_k)$ is also a term of $AY^{-1}(w)$.

On the other hand, assume that
\begin{equation}\label{eq:25}
S\overset{r_{11}}{\Rightarrow}
\alpha_1\overset{r_{12}}{\Rightarrow}\cdots\overset{r_{1l}}{\Rightarrow}wu
\end{equation}
in $G$ for $w=x_1\cdots x_t$ and $u=\sigma_1\cdots\sigma_m$, and there exists states $q_0,q_1,\cdots,q_m=q_f$ in $Q$ such that
\begin{equation}\label{eq:26}
r_{21}=\delta(q_0,\sigma_1,q_1), \cdots, r_{2m}=\delta(q_{m-1},\sigma_m,q_m)\ \mbox{and}\ r_2=r_{21}\otimes\cdots\otimes r_{2m}.
\end{equation}
Then $r=r_{11}\otimes \cdots\otimes  r_{1l}\otimes r_{21}\otimes \cdots\otimes  r_{2m}$ is a term in the sum $AY^{-1}(w)$. Since the induction (\ref{eq:25}) is in $G$, by type (ii-3) productions we have
\begin{eqnarray*}
S'  \hspace*{-3mm}&=& \hspace*{-3mm}(q_0,S,q_f)\overset{r_{11}}{\Rightarrow} \beta_1\overset{r_{12}}{\Rightarrow}\cdots\overset{r_{1l}}{\Rightarrow}  \\
 && \hspace{8mm} (q_0,x_1,q_0)\cdots (q_0,x_m,q_0)(q_0,\sigma,q_1)\cdots (q_{m-1},\sigma_m,q_f),
\end{eqnarray*}
where $q_i$ is chosen as in (\ref{eq:26}) for $i\geq 1$.

Applying type (ii-1) productions to $(q_0,x_i,q_0)$ and type (ii-2) productions to $(q_{i-1}$, $y_i,q_i)$ we see that
\begin{eqnarray*}
S' \hspace*{-5mm}&&\overset{r_{11}}{\Rightarrow}\beta_1\overset{r_{12}}{\Rightarrow}\cdots\overset{r_{1l}}{\Rightarrow} \\
&& (q_0,x_1,q_0)\cdots (q_0,x_m,q_0)(q_0,\sigma,q_1)\cdots (q_{m-1},\sigma_m,q_f)\overset{r_{21}}{\Rightarrow}\cdots
\overset{r_{2m}}{\Rightarrow}w.
\end{eqnarray*}

Hence, $r_{11}\otimes r_{12}\otimes \cdots\otimes  r_{1l}\otimes r_{21}\otimes \cdots\otimes  r_{2m}$ is a term in the sum of $L(G')(w)$.

Therefore, $L(G')(w)=AY^{-1}(w)$ for any $w\in\sa$. Therefore, $L(G')=AY^{-1}$ and $AY^{-1}$ is context-free.

Since $S$ is commutative, by the duality between left and right quotients, it follows that $X^{-1}A$ is also context-free once $X$ is regular and $A$ is context-free.

(iv) By the proof of statement (iii), we have the following simple observation: If ${\cal A}$ is a DWA in the proof of statement (iii), i.e., ${\cal A}$ is a classical deterministic finite automaton with a unique final state. Then the commutativity of $S$ is not necessary in the proof of
the above proposition. This is because, in this case, the weights in type (ii-1) and (ii-2) productions take values $1$, and a production of type (ii-3) commutes
with one of  type (ii-1) or (ii-2). In this case, if we let $Y=|{\cal A}|$, then $AY^{-1}$ is context-free.

For a general finite DWA ${\cal A}=(Q,\Sigma,\delta,q_0, F)$, we write ${\cal A}_q=(Q,\Sigma,\delta,q_0, \{q\})$ for a finite DWA with unique final state $q$ for any state $q$ in $Q$. If we let $Y_q=|{\cal A}_q|$, then $AY_q^{-1}$ is context-free by the above observation. By a simple calculation, we have $Y=|{\cal A}|=\sum_{q\in Q}Y_qF(q)$. Then, by Proposition \ref{pro:quotient-fps} (ii),  we have $$AY^{-1}=A(\sum_{q\in Q}Y_qF(q))^{-1}=\sum_{q\in Q}(AY_q^{-1})F(q).$$ Since $AY_q^{-1}$ is context-free for any $q\in Q$ and the family of weighted context-free language is closed under scalar product and finite sum, it follows that $AY^{-1}$ is context-free.

Similar to the proof of statement (ii), $X^{-1}A$ is context-free if $X$ is DWA-regular and $A$ is context-free.
This shows that the statement (iv) holds.
\end{proof}

\section{Proof of Proposition~\ref{pro:strong homo}}

\begin{proof}
Let us first show the following equality holds,
\begin{equation}
Y_g\rw_{incl} (\bigwedge_{q\in Q} \bigwedge_{\delta(p,\sigma)=q}f(p)\rw Fut_{{\cal A}}(q))=\sigma Y_g\rw_{incl} Y_f.
\end{equation}
Consider $D=\sigma (\bigwedge_{q\in Q}\bigwedge_{\delta(p,\sigma)=q}f(p)\rw Fut_{{\cal A}}(q))$.

If $u\not= \sigma v$ for any $v\in\sa$, then $D(u)=0$.

If $u=\sigma v$ for some $v\in \sa$, then
\begin{eqnarray*}
D(u) &=& \bigwedge_{q\in Q}\bigwedge_{\delta(p,\sigma)=q}f(p)\rw Fut_{{\cal A}}(q)(v)\\
 &=& \bigwedge_{q\in Q}\bigwedge_{\delta(p,\sigma)=q}f(p)\rw F(\ds(q,v)) \\
 &=& \bigwedge_{q\in Q} \bigwedge_{\delta(p,\sigma)=q}f(p)\rw F(\ds(q,\sigma v)) \\
 &=& \bigwedge_{p\in Q}f(p)\rw F(\ds(p,u)) \\
 &=& \bigwedge_{p\in Q}f(p)\rw Fut_{{\cal A}}(p)(u)=Y_f(u).
 \end{eqnarray*}

Noting that if $u$ does not have the form $\sigma v$ for any $v\in \sa$, then $\sigma Y_g(u)=0$. Thus,
\begin{eqnarray*}
&& Y_g\rw_{incl} (\bigwedge_{q\in Q}\bigwedge_{\delta(p,\sigma)=q}f(p)\rw Fut_{{\cal A}}(q))\\
 &=& \sigma Y_g\rw_{incl} \sigma(\bigwedge_{q\in Q}\bigwedge_{\delta(p,\sigma)=q}f(p)\rw Fut_{{\cal A}}(q)) \\
&=& \sigma Y_g\rw_{incl} D=\sigma Y_g\rw_{incl} Y_f.
\end{eqnarray*}
Next, let us prove the following equality
\begin{equation}\label{eq:49}
\sigma Y\rw_{incl} Y_f=\bigvee\{f\sigma\rw_{incl} g | Y_g=Y\}.
\end{equation}
On one hand,  suppose $c\leq f\sigma\rw_{incl} g$ for some $c\in L$ and $Y=Y_g$. Then we have $c\otimes f\sigma (q)\leq g(q)$ for any $q\in Q$. Therefore, $g(q)\rw Fut_{{\cal A}}(q)\leq c\otimes f\sigma (q)\rw Fut_{{\cal A}}(q)=c\rw(f\sigma(q)\rw Fut_{{\cal A}}(q))$ for any $q\in Q$.  Hence,
\begin{eqnarray*}
\bigwedge_{q\in Q} (g(q)\rw Fut_{{\cal A}}(q)) &\leq& \bigwedge_{q\in Q} (c \rw (f\sigma (q)\rw Fut_{{\cal A}}(q)) \\
&=& c\rw(\bigwedge_{q\in Q}f\sigma(q)\rw Fut_{{\cal A}}(q)).
\end{eqnarray*}
This further implies that
\begin{eqnarray*}
c &\leq& (\bigwedge_{q\in Q}g(q)\rw Fut_{{\cal A}}(q))\rw_{incl} (\bigwedge_{q\in Q}f\sigma (q)\rw Fut_{{\cal A}}(q)) \\
&=& Y_g\rw_{incl} (\bigwedge_{q\in Q}\bigwedge_{\delta(p,\sigma)=q}f(p)\rw Fut_{{\cal A}}(q)) \\
&=& \sigma Y_g\rw_{incl} Y_f.
\end{eqnarray*}

This shows that $f\sigma\rw_{incl} g\leq \sigma Y_g\rw_{incl} Y_f$.

On the other hand, let $X'_f(u)=f(\ds(q_0,u))$. By Proposition \ref{prop: Xf&Yf}, we know $X'_f\leq X_f$. Thus,
\begin{equation}
\sigma Y_g\rw_{incl} Y_f=X_f\sigma\rw_{incl} X_g\leq X'_f\sigma\rw_{incl} X_g.
\end{equation}

We next show
\begin{equation}\label{eq:51}
X'_f\sigma\rw_{incl} X_g\leq f\sigma\rw_{incl} g,
\end{equation}
for some weighted state $g$ (in fact, the largest weighted state $g$ such that $Y_g=Y$), where $g$ is defined by, $$g(q)=\bigvee\{X(u) | \ds(q_0,u)=q\}.$$
By Proposition \ref{pro:facor-state1}, $g$ satisfies $Y_g=Y$. Then
$$\sigma Y\rw_{incl} Y_f=\sigma Y_g\rw_{incl} Y_f\leq X'_f\sigma\rw_{incl} X_g\leq f\sigma\rw_{incl} g,$$
and then Eq.(\ref{eq:49}) holds.

The proof of Eq.(\ref{eq:51}) is as follows.

Note that
\begin{eqnarray*}
X'\sigma\rw_{incl} X_g &=& X'\sigma\rw_{incl} X \\
&=& \bigwedge_{u\in \sa}X'\sigma(u)\rw X(u) \\
&=& \bigwedge_{v\in \sa}X'(v)\rw X(v\sigma) \\
&\leq&  \bigwedge_{v\in \sa}f(\ds(q_0,v))\rw g(\ds(q_0,v\sigma)) \\
&=& \bigwedge_{q\in Q}f(q)\rw g(\delta(q,\sigma)).
\end{eqnarray*}
and
\begin{eqnarray*}
f\sigma\rw_{incl} g &=& \bigwedge_{q\in Q}f\sigma(q)\rw g(q) \\
& =& \bigwedge_{q\in Q}(\bigvee\{f(p) | \delta(p,\sigma)=q\})\rw  g(q) \\
& =& \bigwedge_{q\in Q}(\bigwedge_{\delta(p,\sigma)=q}f(p)\rw g(\delta(p,\sigma)) \\
&=& \bigwedge_{q\in Q}f(q)\rw g(\delta(q,\sigma)).
\end{eqnarray*}
We know Eq.(\ref{eq:51}) holds.

We next show that $\varphi$ also satisfies the following two conditions:
\begin{eqnarray}
\label{eq:JA_xy}
J_A(X,Y) &=& \bigvee\{J(f) | Y=Y_f\}, \\
\label{GA-phi}
G_A(\varphi(f)) &=& G(f).
\end{eqnarray}
For Eq.(\ref{eq:JA_xy}), we have  $J_A(X,Y)=X(\varepsilon)$, and $Y=Y_f$, then $J(f)=f(q_0)=f(\ds(q_0,\varepsilon))=X'(\varepsilon)\leq X(\varepsilon)$. If we let $X_f=X$, and thus $J(f)$
can take $X_f(\varepsilon)$, this shows that Eq.(\ref{eq:JA_xy}) holds.

For Eq.(\ref{GA-phi}), we have
\begin{eqnarray*}
G_A(\varphi(f)) = G_A(X_f,Y_f) &=& Y_f(\varepsilon)\\ &=&\bigwedge_{q\in Q}f(q)\rw Fut_{{\cal A}}(q)(\varepsilon) \\
 &=& \bigwedge_{q\in Q}f(q)\rw F(\ds(q,\varepsilon) \\
 &=& \bigwedge_{q\in Q}f(q)\rw F(q)\\
 &=& f\rw_{incl} F=G(f).
\end{eqnarray*}
Eq.(\ref{eq:49}), Eq.(\ref{eq:JA_xy}), and Eq.(\ref{GA-phi}) imply that $\varphi$ is a strong homomorphism from ${\cal A}_1$ onto ${\cal U}_A$, and then $|{\cal A}_1|=|{\cal U}_A|=A$.
\end{proof}

\end{document}